\PassOptionsToPackage{dvipsnames}{xcolor}
\pdfoutput=1

\newif\ifdraft\draftfalse 
\newif\ifanon\anonfalse    
\newif\iffull\fullfalse   
\newif\ifdiff\difffalse 
\newif\iflongrefs\longrefsfalse 
\newif\ifbackref\backreffalse 
\newif\ifsooner\soonerfalse
\newif\iflater\laterfalse
\newif\ifieee\ieeetrue
\newif\ifcamera\cameratrue 
\newif\ifappendix\appendixfalse 
\newif\ifallcites\allcitesfalse
\newif\ifneedspace\needspacetrue
\newif\ifhighlightnewtext\highlightnewtexttrue 
\newif\ifhighlightnewreviews\highlightnewreviewsfalse 
\newif\ifaggressivecomments\aggressivecommentsfalse 
\makeatletter
  \@input{texdirectives.tex} 
\makeatother

\ifieee
\documentclass[10pt, conference, compsocconf, letterpaper, times]{IEEEtran}

\DeclareMathAlphabet{\mathit}{\encodingdefault}{\familydefault}{m}{it}

\else 

\ifcamera
\documentclass[acmsmall]{acmart}
\else
\ifanon
\documentclass[acmsmall,review,anonymous,nonacm]{acmart}
\else
\documentclass[acmsmall,review=false,screen,nonacm]{acmart}
\fi
\fi

\acmJournal{PACMPL}
\acmYear{2018}
\acmMonth{1}
\acmVolume{2}
\aacmNumber{POPL}
\acmPrice{}
\acmArticle{65}
\acmMonth{1}
\acmDOI{10.1145/3158153}
\ifcamera\else\copyrightyear{2017}\fi 

\ifanon
\setcopyright{none}             
\else
\setcopyright{rightsretained}   
\fi

\makeatletter
\def\@copyrightpermission{\ifcamera\\\\\\\fi This work is licensed under a \href{https://creativecommons.org/licenses/by/4.0/}{Creative Commons Attribution 4.0 International License}}
\makeatother

\fi 

\usepackage{booktabs}   
\usepackage{subcaption} 

\usepackage[inference]{semantic}
\usepackage{amsmath}\allowdisplaybreaks
\usepackage{mathpartir}
\mprset {sep=1em} 
\usepackage{amsthm}
\usepackage{amssymb}
\usepackage{stmaryrd} 
\usepackage{xspace}
\usepackage{latexsym}
\usepackage{ifthen}
\usepackage{mathtools}
\usepackage{color}
\usepackage{listings}
\usepackage{verbatim}
\usepackage{tikz}
\usetikzlibrary{positioning,shadows,arrows,calc,backgrounds,fit,shapes,shapes.multipart,decorations.pathreplacing,shapes.misc,patterns,decorations.markings}
\usepackage{tikzscale}
\usepackage{tikz-qtree}
\usepackage[T1]{fontenc}
\usepackage[scaled=.83]{beramono}
\usepackage{epigraph}
\usepackage{booktabs}
\usepackage{float}
\usepackage{etoolbox}
\usepackage{wrapfig}
\usepackage{scalerel}
\usepackage{nameref}
\usepackage[strict]{changepage}
\usepackage{bm}
\usepackage{graphicx}
\usepackage[makeroom]{cancel}
\usepackage{centernot}
\usepackage{thm-restate}
\usepackage{enumitem}
\usepackage{fancyvrb}
\usepackage{varwidth}
\usepackage{combelow}

\usepackage{titlecaps}
\usepackage[normalem]{ulem}
\usepackage{cals}
\usepackage{afterpage}
\usepackage{placeins}

\ifieee
\usepackage[numbers,sort]{natbib}
\newcommand\citepos[1]{\citeauthor{#1}'s\ \cite{#1}}
\fi

\usepackage{chngcntr}

\ifanon
\usepackage[switch]{lineno}
\linenumbers

\fi

\usepackage{stfloats}

\definecolor{dkblue}{rgb}{0,0.1,0.5}
\definecolor{dkgreen}{rgb}{0,0.3,0}
\definecolor{dkred}{rgb}{0.6,0,0}
\definecolor{dkpurple}{rgb}{0.7,0,1.0}
\definecolor{purple}{rgb}{0.9,0,1.0}
\definecolor{olive}{rgb}{0.4, 0.4, 0.0}
\definecolor{teal}{rgb}{0.0,0.4,0.4}
\definecolor{azure}{rgb}{0.0, 0.5, 1.0}
\definecolor{gray}{rgb}{0.5, 0.5, 0.5}
\definecolor{dkgray}{rgb}{0.3, 0.3, 0.3}

\definecolor{cbgreen}{RGB}{000,158,115}
\definecolor{cbblue}{RGB}{000,114,178}
\definecolor{cbteal}{RGB}{091,142,253}
\definecolor{cbyellow}{RGB}{240,228,066}
\definecolor{cborange}{RGB}{230,140,000}
\definecolor{cbred}{RGB}{213,094,000}
\definecolor{cbpurple}{RGB}{204,121,167}

\newcommand*{\colSel}[1]{{\color{cbblue} #1}}
\newcommand*{\colU}[1]{{\color{cbred} #1}}

\newcommand*{\colFlexnew}[1]{{\color{cbpurple} #1}}
\newcommand*{\colAll}[1]{{\color{black} #1}}

\ifieee
\usepackage[
    pdftex,%
    pdfpagelabels,%
    colorlinks,%
    linkcolor=dkblue,%
    citecolor=dkblue,%
    filecolor=dkblue,%
    urlcolor=dkblue%
    \ifbackref,backref\fi%
]{hyperref}
\else
\usepackage{hyperref}
\hypersetup{
  breaklinks=true,
  citecolor=black,
  linkcolor=black,
  urlcolor=black,
}
\fi

\def\Snospace~{\S{}}

\def\Nnospace~{}

\usepackage{etoolbox}
\makeatletter
\patchcmd{\hyper@makecurrent}{%
    \ifx\Hy@param\Hy@chapterstring
        \let\Hy@param\Hy@chapapp
    \fi
}{%
    \iftoggle{inappendix}{
        \@checkappendixparam{chapter}%
        \@checkappendixparam{section}%
        \@checkappendixparam{subsection}%
        \@checkappendixparam{subsubsection}%
        \@checkappendixparam{paragraph}%
        \@checkappendixparam{subparagraph}%
    }{}%
}{}{\errmessage{failed to patch}}

\newcommand*{\@checkappendixparam}[1]{%
    \def\@checkappendixparamtmp{#1}%
    \ifx\Hy@param\@checkappendixparamtmp
        \let\Hy@param\Hy@appendixstring
    \fi
}
\makeatletter

\newtoggle{inappendix}
\togglefalse{inappendix}

\apptocmd{\appendix}{\toggletrue{inappendix}}{}{\errmessage{failed to patch}}


\AtBeginEnvironment{tikzpicture}{\catcode`$=3 } 

\newcommand{\comm}[3]{\ifdraft{{\color{#1}\ifaggressivecomments\bfseries \color{red}\fi[#2: #3]}}\fi}
\newcommand{\ch}[1]{\comm{teal}{CH}{#1}}
\newcommand{\rb}[1]{\comm{orange}{RB}{#1}}

\newcommand{\jb}[1]{\comm{olive}{JB}{#1}}

\newcommand*{\EG}{e.g.,\xspace}
\newcommand*{\IE}{i.e.,\xspace}

\newcommand*{\SLH}{SLH\xspace}
\newcommand*{\iSLH}{i\SLH}

\newcommand*{\vSLH}{v\SLH}
\newcommand*{\SelectiveSLH}{Selective \SLH}

\newcommand*{\SelectivevSLH}{Selective v\SLH}
\newcommand*{\SelSLH}{SSLH\xspace} 
\newcommand*{\SeliSLH}{SiSLH\xspace}

\newcommand*{\SelvSLH}{SvSLH\xspace}
\newcommand*{\UltimateSLH}{Ultimate \SLH}
\newcommand*{\USLH}{USLH\xspace}
\newcommand*{\FlexibleSLH}{Flexible \SLH}
\newcommand*{\FlexSLH}{FSLH\xspace}
\newcommand*{\FlexiSLH}{FiSLH\xspace}

\newcommand*{\FlexvSLH}{FvSLH\xspace}
\newcommand*{\FsFlexvSLH}{\texorpdfstring{FvSLH$^\forall$}{FvSLH∀}\xspace} 
\newcommand*{\mathFsFlexvSLH}{FvSLH^{\forall}}

\newcommand*{\SourceLang}{\textsc{AWhile}\xspace}

\newcommand*{\CCT}{CCT\xspace}
\newcommand*{\ObsEq}{observational equivalence\xspace}

\newcommand*{\BoolTrue}{\mathbb{T}}
\newcommand*{\BoolFalse}{\mathbb{F}}
\newcommand*{\LabelSecret}{\BoolFalse}
\newcommand*{\LabelPublic}{\BoolTrue}

\newcommand*{\PubVars}{\textit{P}}
\newcommand*{\PubArrs}{\textit{PA}}
\newcommand*{\mathpc}{\textit{pc}}

\newcommand*{\cexpr}[1]{\mathit{#1}}
\newcommand*{\cvar}[1]{\texttt{#1}}

\newcommand*{\annotated}[1]{\overline{#1}}

\newcommand*{\cskip}{\texttt{skip}}
\newcommand*{\casgn}[2]{#1 ~ \texttt{:=} ~ #2}
\newcommand*{\cseq}[2]{#1 \texttt{;} ~ #2}
\newcommand*{\caread}[3]{#1 \leftarrow #2 \texttt{[} #3 \texttt{]}}
\newcommand*{\cawrite}[3]{#1 \texttt{[} #2 \texttt{]} \leftarrow #3}

\newcommand*{\acseq}[4]{#1 \texttt{;}_{@( #3 , #4 )} ~ #2}
\newcommand*{\acaread}[5]{#1_{@ #4} \leftarrow #2 \texttt{[} #3_{@ #5}\texttt{]}}
\newcommand*{\acawrite}[4]{#1 \texttt{[} #2_{@ #4} \texttt{]} \leftarrow #3}
\newcommand*{\acbranch}[2]{\texttt{branch} ~ #1 ~ #2}

\newcommand*{\cifs}[3]{\texttt{if} ~ #1 ~ \texttt{then} ~ #2 ~ \texttt{else} ~ #3}
\newcommand*{\cifsnoelse}[2]{\texttt{if} ~ #1 ~ \texttt{then} ~ #2}
\newcommand*{\cwhiles}[2]{\texttt{while} ~ #1 ~ \texttt{do} ~ #2}
\newcommand*{\acifs}[4]{\texttt{if} ~ #1_{@ #4} ~ \texttt{then} ~ #2 ~ \texttt{else} ~ #3}
\newcommand*{\acwhiles}[5]{\texttt{while} ~ #1_{@ #3} ~ \texttt{do} ~ #2_{@ (#4, #5)}}

\newcommand*{\ccond}[3]{#1 ~ \texttt{?} ~ #2 ~ \texttt{:} ~ #3}
\newcommand*{\ccondtight}[3]{#1 \texttt{?} #2 \texttt{:} #3}

\newcommand*{\wtifcAux}[4]{#1; #2 \vdash_{#3} #4}
\newcommand*{\wtifc}[2]{\wtifcAux{\PubVars}{\PubArrs}{#1}{#2}}
\newcommand*{\wtct}[1]{\PubVars; \PubArrs \vdash #1}

\newcommand*{\wlacom}[6]{#2 , #3 \leadsto #5, #6 \vdash_{#4} #1}
\newcommand*{\pcofacom}[2]{\textit{pc-after}\, #1 \, #2}
\newcommand*{\branchfree}[1]{\textit{branch-free}\, #1}

\newcommand*{\transl}[1]{\llparenthesis ~ #1 ~ \rrparenthesis}
\newcommand*{\translAux}[3]{\transl{#1}^{\mathit{#2}}_{#3}}
\newcommand*{\transls}[2]{\translAux{#1}{#2}{\PubVars}}
\newcommand*{\translFsFlexvSLH}[1]{\translAux{#1}{\mathFsFlexvSLH}{}}

\newcommand*{\flowtrack}[4]{\langle\!\langle #1 \rangle\!\rangle^{#2,#3}_{#4}}

\newcommand*{\eval}[2]{\llbracket #1 \rrbracket_{#2}}

\newcommand*{\seqstate}[3]{\langle #1,#2,#3\rangle}
\newcommand*{\specstate}[4]{\langle #1,#2,#3,#4\rangle}
\newcommand*{\SeqEval}[7]{\seqstate{#1}{#2}{#3} \seqstep{#7} \seqstate{#4}{#5}{#6}}
\newcommand*{\SpecEval}[1]{
  \def\SpecEvalTemp{{#1}}
  \SpecEvalAux
}
\newcommand*{\SpecEvalAux}[9]{\specstate{\SpecEvalTemp}{#1}{#2}{#3} \specstep{#8}{#9} \specstate{#4}{#5}{#6}{#7}}

\newcommand*{\IdealEval}[1]{
  \def\IdealEvalTemp{{#1}}
  \IdealEvalAux
}
\newcommand*{\IdealEvalAux}[9]{\idealstep{\specstate{\IdealEvalTemp}{#1}{#2}{#3}}{#8}{#9}{\specstate{#4}{#5}{#6}{#7}}}

\newcommand*{\fsidealstate}[7]{\langle #1, #2, #3, #4, #5, #6, #7 \rangle}

\newcommand*{\step}{\textit{step}}
\newcommand*{\branch}[1]{\textit{branch}\,#1}
\newcommand*{\force}{\textit{force}}
\newcommand*{\OARead}[2]{\textit{read}\,#1\,#2}
\newcommand*{\OAWrite}[2]{\textit{write}\,#1\,#2}
\newcommand*{\DLoad}[2]{\textit{load}\,#1\,#2}
\newcommand*{\DStore}[2]{\textit{store}\,#1\,#2}

\newcommand*{\ScalarState}{\rho}
\newcommand*{\ArrayState}{\mu}
\newcommand*{\SemSpecFlag}{b}
\newcommand*{\Obss}{\mathcal{O}}
\newcommand*{\Dirs}{\mathcal{D}}

\newcommand*{\subst}[3]{[#1 \mapsto #2] #3}
\newcommand*{\lookup}[2]{#1 ( #2 )}
\newcommand*{\UsedVars}[1]{\mathit{VARS} ( #1 )}

\newcommand*{\natofbool}[1]{\llbracket #1 \rrbracket_{\mathbb{N}}}

\newcommand*{\seqstep}[1]{\xrightarrow{#1}}
\newcommand*{\seqmulti}[1]{\xrightarrow{#1}{\hspace{-0.5em}}{}^*}
\newcommand*{\specstep}[2]{\xrightarrow[#2]{#1}\hspace{-0.5em}{}_s}
\newcommand*{\specmulti}[2]{\xrightarrow[#2]{#1}{\hspace{-0.5em}}{}_s^*}
\newcommand*{\idealsteparrow}[2]{\xrightarrow[#2]{#1}{\hspace{-0.5em}}{}_i}
\newcommand*{\idealmultiarrow}[2]{\xrightarrow[#2]{#1}{\hspace{-0.5em}}{}_i^*}

\newcommand{\specobseq}{\approx_\mathit{s}}
\newcommand{\seqobseq}{\approx}

\newcommand*{\idealstep}[4]{#1 \xrightarrow[#3]{#2}\hspace{-0.5em}{}_i #4}
\newcommand*{\idealmulti}[4]{#1 \xrightarrow[#3]{#2}\hspace{-0.5em}{}_i^* #4}

\newcommand*{\seqmultis}[3]{#1 \seqmulti{#2} #3}

\newcommand*{\labelarith}[2]{#1(#2)} 
\newcommand*{\labelbool}[2]{#1(#2)} 
\newcommand*{\ArithLabel}[1]{\PubVars(#1)}
\newcommand*{\BoolLabel}[1]{\PubVars(#1)}

\newcommand*{\ValueCheck}[2]{\mathit{VC}(#1, #2)}

\renewcommand*{\eqref}[1]{(\ref{#1})}

\newcommand{\newreviews}[1]{\ifhighlightnewreviews{\color{dkblue}#1}\else#1\fi}

\newcommand*{\highlight}[1]{{\color{cbred} #1}}
\newcommand*{\colHighlight}[1]{{\color{cbred} #1}}

\theoremstyle{definition}
\newtheorem{example}{Example}
\newtheorem{definition}{Definition}
\newtheorem{theorem}{Theorem}
\newtheorem{lemma}{Lemma}

\newtheorem{listing}{Listing}

\ifieee\else\fi

\newcommand{\titleString}{FSLH: Flexible Mechanized Speculative Load Hardening}
\title{\titleString}

\ifneedspace
\title{\huge\bf{\titleString}
\ifneedspace\ifieee\ifanon\vspace{-1.5em}\else\vspace{-0.5em}\fi\fi\else\ifanon\vspace{2em}\fi\fi}
\else
\title{\titleString}
\subtitle{\subtitleString}
\fi

\ifanon
\author{}
\else

\ifieee\large
\author{
  Jonathan Baumann\textsuperscript{1,2} \quad 
  Roberto Blanco\textsuperscript{1,3} \quad 
  L\'eon Ducruet\textsuperscript{1,4} \quad
  Sebastian Harwig\textsuperscript{1,5} \quad 
  C\u{a}t\u{a}lin Hri\cb{t}cu\textsuperscript{1} \quad 
\\[0.5em]
{\small
  \textsuperscript{1}MPI-SP,\ifcamera\else{} Bochum,\fi{} Germany\quad
  \textsuperscript{2}ENS Paris-Saclay, France\quad
  \textsuperscript{3}TU/e, Netherlands\quad
  \textsuperscript{4}ENS Lyon, France\quad
  \textsuperscript{5}Ruhr University Bochum, Germany}\\[0em]
}
\else 

\author{AUTHOR1}
\affiliation{
  \ifcamera\institution{AFF1}\city{CITY}\country{COUNTRY}
  \else\institution{AFF1}\fi}
\email{EMAIL}

\makeatletter
\renewcommand{\@shortauthors}{SHORTAUTHORS}
\makeatother

\fi 

\fi 

\begin{document}


\ifieee\maketitle\fi


\begin{abstract}
The Spectre speculative side-channel attacks pose formidable threats for security.
Research has shown that code following the cryptographic constant-time discipline can be efficiently protected against Spectre v1 using a selective variant of Speculative Load Hardening (SLH).
SLH was, however, not strong enough for protecting non-cryptographic code, leading to the introduction of Ultimate SLH, which provides protection for arbitrary programs, but has too large overhead for general use, since it conservatively assumes that all data is secret.
In this paper we introduce a flexible SLH notion that achieves the best of both worlds by generalizing both Selective and Ultimate SLH.
We give a suitable security definition for such transformations protecting arbitrary programs: any transformed program running with speculation should not leak more than what the source program leaks sequentially.
We formally prove using the Rocq prover that two flexible SLH variants enforce this relative security guarantee.
As easy corollaries we also obtain that, in our setting, Ultimate SLH enforces our relative security notion, and two selective SLH variants enforce speculative constant-time security.

\end{abstract}

\ifieee
\renewcommand\IEEEkeywordsname{Keywords}
\begin{IEEEkeywords}
side-channel attacks, speculative execution, Spectre, secure compilation,
speculative load hardening, speculative constant time, relative security,
formal verification, Rocq, Coq
\end{IEEEkeywords}
\fi

\section{Introduction}
\label{sec:intro}

%
%
%
%
Speculative side-channel attacks such as Spectre pose formidable threats for the
security of computer systems~\cite{CanellaBSLBOPEG19, McIlroySTTV19}.
%
%
For instance, in typical Spectre v1 attacks~\cite{KocherHFGGHHLM019},
misspeculated array bounds checks cause out-of-bounds memory accesses to
inadvertently load and reveal secrets via timing variations.
SLH~\cite{Carruth18} is a software countermeasure against such attacks
that dynamically tracks whether execution is in a
mispredicted branch using a misspeculation flag register.

It is, however, challenging to build software protections
that are both efficient and
that provide formal end-to-end security guarantees
against precisely specified, speculative side-channel attacker
models~\cite{CauligiDMBS22}. 
Cryptography researchers are leading the way in this space, with defenses such
as {\em selective} SLH efficiently achieving speculative constant-time
guarantees against Spectre v1 for cryptographic code with typical overheads
under 1\%~\cite{ShivakumarBBCCGOSSY23, ShivakumarBGLOPST23}.
%
This work is, however, specialized to only cryptographic code and
often also to domain-specific languages for cryptography,
such as Jasmin~\cite{AlmeidaBBBGLOPS17}.

It is more difficult to properly protect arbitrary programs written in
general-purpose languages, which in particular do not obey the cryptographic
constant-time discipline~\cite{AlmeidaBBDE16,CauligiSJBWRGBJ19}.
SLH was not strong enough for protecting such non-cryptographic
code~\cite{PatrignaniG21}, leading to the introduction of Ultimate
SLH~\cite{PatrignaniG21, ZhangBCSY23}, which uses the misspeculation flag
to mask not only the values loaded from memory, but also
all branch conditions, all memory addresses,
the operands of all non-constant time operations, etc.
While this should in principle be strong enough to achieve a relative security
notion~\cite{CauligiDMBS22, ZhangBCSY23, PatrignaniG21},
it also brings $\sim$150\% overhead on the SPEC benchmarks~\cite{ZhangBCSY23},
which seems unacceptable for many practical scenarios.

In this paper we introduce FSLH, a flexible SLH notion that achieves the best of
both worlds by generalizing both Selective and Ultimate SLH.
Like Selective SLH, FSLH keeps track of which program inputs are secret and
which ones not and only protects those memory operations that could potentially
leak secret inputs when ran speculatively.
Like Ultimate SLH, FSLH also provides protection for non-cryptographic code that
does not respect the constant-time discipline, but does this by only masking
those observable expressions 
(\EG branch conditions)
that are potentially influenced by secrets.

\paragraph{Contributions}
\begin{itemize}[leftmargin=*,nosep,label=$\blacktriangleright$]
\item We introduce FSLH, a Flexible SLH notion generalizing both
  Selective SLH and Ultimate SLH.
  We investigate this idea in the context of a simple while language
  with arrays by studying two variants: (1) {\em Flexible value SLH} (FvSLH),
  which uses the misspeculation flag as a mask to erase the
  {\em value} of misspeculated array loads~\cite{ShivakumarBBCCGOSSY23, ShivakumarBGLOPST23}; and (2)
  {\em Flexible index SLH} (FiSLH), which masks {\em array indices}.
%
\item We prove formally in the Rocq prover\footnote{The
  \href{https://rocq-prover.org/}{Rocq interactive theorem prover} was previously
  known as Coq.} that FvSLH and FiSLH can each be instantiated to recover both the
  corresponding Selective SLH variant and Ultimate SLH.
  The connection to Selective SLH guarantees that code following the constant-time
  discipline pays no extra performance penalty under Flexible SLH compared to Selective SLH.
  The connection to Ultimate SLH states that, in our simple setting, if all data is labeled as secret
  then the program receives the same protection as with Ultimate SLH.

\item We give a suitable relative security definition for transformations
  protecting arbitrary programs, like FSLH and Ultimate SLH: any
  transformed program running with speculation should not leak more than what
  the source program leaks sequentially (so without speculation).

\item We prove in Rocq that our two flexible SLH variants, FvSLH and FiSLH,
  enforce this relative security notion.

\item As a corollary we obtain in Rocq that, in our setting, Ultimate SLH also enforces
  relative security, which is the first machine-checked security proof for
  Ultimate SLH~\cite{PatrignaniG21,ZhangBCSY23}.


\item Other easy corollaries we obtain in Rocq are that Selective vSLH and
  Selective iSLH enforce speculative constant-time security, whereas previously
  only Selective vSLH was proved secure and only on
  paper~\cite{ShivakumarBBCCGOSSY23}.
  For Selective iSLH we give a counterexample showing that masking the
  indices of array stores of secret values is required for security.

\item For simplicity, the two FSLH variants above make use of a
  standard information flow control (IFC) type system~\cite{VolpanoIS96}
  with a fixed assignment of variables as public or secret and that rejects some
  source programs as ill-typed.  We show that this is not a
  conceptual limitation by devising \FsFlexvSLH, which instead of the type
  system uses a flow-sensitive IFC analysis that accepts arbitrary source
  programs~\cite{HuntS06}. We prove in Rocq that \FsFlexvSLH enforces relative security {\em
    for all} source programs, the same as Ultimate SLH.
\end{itemize}

\paragraph{Outline} \autoref{sec:background} presents required background,
after which \autoref{sec:key-ideas} introduces our key ideas on FiSLH.
\autoref{sec:defs} defines our simple imperative language and its sequential and
speculative semantics.
\autoref{sec:formal-results} describes the proofs of our formal results for FiSLH.
\autoref{sec:FvSLH} outlines the analogous results for FvSLH,
while \autoref{sec:fsfvslh} presents the more interesting proofs for \FsFlexvSLH.
\autoref{sec:related-work} describes related work and 
\autoref{sec:future-work} concludes with future work.
\ifappendix
The appendices include additional technical details.
\else
An extended version of the paper includes appendices with additional technical
details and is available at \url{https://arxiv.org/abs/2502.03203}
\fi

\paragraph{Artifact} All the results of this paper have been fully formalized in
the Rocq proof assistant.
The Rocq development leading to our main theorems has $\sim$4300 lines of code
and is available at \url{https://github.com/secure-compilation/fslh-rocq}

\section{Background}
\label{sec:background}

\subsection{Speculative Execution Attacks}
\label{sec:spectre}

Modern processors implement a variety of hardware mechanisms, such as caches, to
get software to run faster.
Caches are part of the internal
microarchitectural state of the CPU: they improve performance, but are
not exposed to the programmer through the instruction set architecture (ISA) and
thus have no bearing on the result of any computations.
%
Yet their effects on performance can be detected by measuring the execution time
of a program, which may leak information about its data.
Timing side-channel attacks exploiting
microarchitectural leakage are a common concern for cryptographic code
and other secret-manipulating programs.
%



\emph{Speculative execution} is another technique designed to increase
performance by executing instructions
that are likely, even if not certain, to be needed by a program. For
example, the condition of a branch instruction could depend on slow memory
accesses or the results of instructions that have not finished executing.
Instead of waiting, a processor can use a branch predictor to
guess an outcome for the condition and immediately start executing
instructions from the chosen branch. When the result of the condition
is later known, the processor commits the effects of the speculative
path if the guess was correct, or reverts them and starts executing
the correct path if the guess was wrong.
This rollback ensures that misspeculation has no effect
on 
ISA state, but does not undo changes to the microarchitectural state, like caches.

This opens the door to \emph{speculative leaks} via timing side
channels that would not exist if there was no misspeculation. Even
more, 
an attacker can train the predictors to misspeculate (directly if collocated, or
even indirectly by causing the program's own code to be cleverly
invoked~\cite{SchwarzSLMG19}), which is the essence of Spectre attacks.
A Spectre v1 attack~\cite{KocherHFGGHHLM019}
(aka Spectre-PHT~\cite{CanellaBSLBOPEG19})
targets the branch predictor
to steer a target program down certain speculative paths that
create timing leakage.
%
This is the classic speculative execution attack:



\begin{listing}[Spectre v1 gadget]
\label{ex:v1std}

\begin{align*}
&~\cifsnoelse{\cvar{i < a$_1$\_size}}
  {\\&\quad~\cseq
   {\caread{\cvar{j}}{\cvar{a}_1}{\cvar{i}}}
   {\\&\quad~ \caread{\cvar{x}}{\cvar{a}_2}{\cvar{j}}&}
  }
\end{align*}
\end{listing}



This code indexes a public array $\cvar{a}_1$ over an untrusted input
$\cvar{i}$. A conditional validates that the input is in the right
range, and only in this case accesses $\cvar{a}_1$ at that position. It
then uses the result to index a read from a second array $\cvar{a}_2$.
The problem is that prior to invoking this code, an attacker can train
the branch predictor to guess that the \texttt{then} branch will be
taken. The attacker can then run the program with an out-of-bounds input $\cvar{i}$
that actually points at secret data, and
although the \texttt{then} branch should not be taken, the processor can
start running the instructions on the \texttt{then} branch before the
actual value of the branch condition is resolved. The first array
access loads out-of-bounds secret data into $\cvar{j}$, and
the second access loads from $\cvar{a}_2$ at an address that depends
on\jb{not sure if we should leave this, or change it to "at index j"}\ch{address seems fine to me here}
that $\cvar{j}$.
While the processor will course-correct as soon as it detects its
misspeculation, the access to $\cvar{a}_2\cvar{[j]}$ leaves a footprint in the
cache, so the attacker can measure timing to infer the value of $\cvar{j}$,
which contains speculatively loaded secret data.

\subsection{Speculative Load Hardening}

Speculative Load Hardening~\cite{Carruth18} (\SLH) is a family of software-based
countermeasures against Spectre v1 attacks.
\SLH is a program transformation that combines two ingredients:

\begin{enumerate}

\item a mechanism to keep track of misspeculation, and

\item a means of preventing speculative leaks.

\end{enumerate}

The first is shared by all \SLH variants. The core transformation
from \autoref{fig:aslh-template} shows how
to maintain a \emph{misspeculation flag} tracking whether
any branch was mispredicted along the current control path at any given time.
This is possible because all modern architectures have support for
implementing control flow conditions using branchless and unpredicted
instructions, \EG conditional moves in x86 assembly.

We can easily implement such a flag in a language that uses those
instructions to provide a constant-time conditional, like the
$\cexpr{be} \texttt{ ? } \cexpr{e_1} \texttt{ : } \cexpr{e_2}$ operator
in \autoref{fig:aslh-template}, which evaluates $\cexpr{e_1}$ if
$\cexpr{be}$ is true, and $\cexpr{e_2}$ if it is false.
%
We assume the variable $\cvar{b}$ is reserved for the flag and unused by
programs, and we ensure that its value is 1 if the processor is
misspeculating, and 0 if it is not.
%
The transformation applies recursively to the structure of
commands. Ignoring for now the helpers that parameterize the
scheme ($\llbracket \cdot \rrbracket_{\mathbb{B}}$,
$\llbracket \cdot \rrbracket_{\mathit{rd}}$ and
$\llbracket \cdot \rrbracket_{\mathit{wr}}$), the interesting cases
are the two branching commands: conditionals and loops.

In both cases there are two paths: one where the branch condition is
true 
and one where it is false.
%
For instance,
if the branch condition is true, the command $\cexpr{c_1}$ in the \text{then}
branch (resp.  in the loop body)
works assuming that $\llbracket \cexpr{be} \rrbracket_{\mathbb{B}}$ is true. We use a
constant-time conditional
to evaluate the same condition and
determine if the processor is misspeculating. If it is not, the
condition will indeed evaluate to true, and we will leave the
speculation flag unchanged (\IE no new misspeculation has
occurred). If it is misspeculating, the condition will evaluate to false, and we will
set the flag to 1. This allows the program to
detect misspeculation and protect itself accordingly.

The second ingredient of the transformation is using the misspeculation
flag to prevent speculative leaks.
\autoref{fig:aslh-template} provides
a generic template for index SLH,\footnote{We choose the name {\em index} SLH to
  distinguish it from SLH variants that protect lower-level {\em
    addresses}~\cite{ZhangBCSY23, Carruth18}, instead of high-level {\em array indices}.}
where the indices of array
loads and stores, as well as the branch conditions can be sanitized in
different ways, resulting in different variants, each
with different trade-offs between their degree of protection against
speculative side channels and their performance cost.

Let's first consider a basic \iSLH transformation, which
leaves branch conditions unchanged, \IE
$\llbracket \cexpr{be} \rrbracket_{\mathbb{B}} \doteq \cexpr{be}$, and
protects the indices of all array accesses based on the speculation flag.
If $\cvar{b}$ contains the speculation flag and $\cvar{i}$ is the
index we intend to access,
we can use the constant-time conditional expression
$\ccond{\cvar{b==1}}{\cvar{0}}{\cvar{i}}$ to protect it: when the
processor is misspeculating (as indicated by the set flag $\cvar{b==1}$) the index is
zeroed and the access protected;
when the processor does not misspeculate, the access is left unchanged.
That is, $\llbracket \cexpr{i} \rrbracket_{\mathit{rd}}^{\cvar{X}}
\doteq \llbracket \cexpr{i} \rrbracket_{\mathit{wr}}^{\cexpr{e}}
\doteq \ccond{\cvar{b==1}}{\cvar{0}}{\cexpr{i}}$.
%
%

\begin{figure}
\centering
\begin{align*}
\transl{\cskip} &\doteq \cskip \\
\transl{\casgn{\cvar{x}}{\cexpr{e}}} &\doteq \casgn{\cvar{x}}{\cexpr{e}} \\
\transl{\cseq{\cexpr{c_1}}{\cexpr{c_2}}} &\doteq \cseq{\transl{\cexpr{c_1}}}{\transl{\cexpr{c_2}}} \\
\transl{\cifs{\cexpr{be}}{\cexpr{c_1}}{\cexpr{c_2}}} &\doteq \texttt{if }{\llbracket \cexpr{be} \rrbracket_{\mathbb{B}}} \\
\phantom{\transl{\cskip}} &\phantom{\doteq}\hspace{0.5em} \texttt{then }{\cseq{\casgn{\cvar{b}}{\ccondtight{\llbracket \cexpr{be} \rrbracket_{\mathbb{B}}}{\cvar{b}}{\cvar{1}}}}{\transl{\cexpr{c_1}}}} \\
\phantom{\transl{\cskip}} &\phantom{\doteq}\hspace{0.5em} \texttt{else }{\cseq{\casgn{\cvar{b}}{\ccondtight{\llbracket \cexpr{be} \rrbracket_{\mathbb{B}}}{\cvar{1}}{\cvar{b}}}}{\transl{\cexpr{c_2}}}} \\
\transl{\cwhiles{\cexpr{be}}{\cexpr{c}}} &\doteq \texttt{while }{\llbracket \cexpr{be} \rrbracket_{\mathbb{B}}}\texttt{ do} \\
\phantom{\transl{\cskip}} &\phantom{\doteq}\hspace{1.5em} {\cseq{\casgn{\cvar{b}}{\ccondtight{\llbracket \cexpr{be} \rrbracket_{\mathbb{B}}}{\cvar{b}}{\cvar{1}}}}{\transl{\cexpr{c}}}} \\
\phantom{\transl{\cskip}} &\phantom{\doteq}\hspace{0.5em} {\casgn{\cvar{b}}{\ccondtight{\llbracket \cexpr{be} \rrbracket_{\mathbb{B}}}{\cvar{1}}{\cvar{b}}}} \\
\transl{\caread{\cvar{X}}{\cvar{a}}{\cexpr{i}}} &\doteq \caread{\cvar{X}}{\cvar{a}}{\llbracket \cexpr{i} \rrbracket_{\mathit{rd}}^{\cvar{X}}} \\
\transl{\cawrite{\cvar{a}}{\cexpr{i}}{\cexpr{e}}} &\doteq \cawrite{\cvar{a}}{\llbracket \cexpr{i} \rrbracket_{\mathit{wr}}^{\cexpr{e}}}{\cexpr{e}}
\end{align*}
\vspace{-1.5em}
\caption{Master recipe for index \SLH}
\label{fig:aslh-template}
\end{figure}

\begin{figure*}[!b]
\subcaptionbox{Selective \iSLH \label{fig:saslh-instance}}[0.38\linewidth]{
\begin{align*}
\llbracket \cexpr{be} \rrbracket_{\mathbb{B}}
 &\doteq \colSel{
  \phantom{\left\{ \right.}
  \begin{matrix*}[l]
  \phantom{\ccond{\cvar{b==1}}{\cvar{0}}{\cexpr{i}}}
  \\
  \cexpr{be}
  \end{matrix*}
  }
\\
\llbracket \cexpr{i} \rrbracket_{\mathit{rd}}^{\cvar{X}}
 &\doteq \colSel{
  \left\{
  \begin{matrix*}[l]
  {\colAll{\ccond{\cvar{b==1}}{\cvar{0}}{\cexpr{i}}}} &
  {\mathit{if} ~ \PubVars(\cvar{X})}
  \\
  {\cexpr{i}} &
  {\mathit{otherwise}}
  \end{matrix*}
  \right.
  }
\\
\llbracket \cexpr{i} \rrbracket_{\mathit{wr}}^{\cexpr{e}}
 &\doteq \colSel{
  \left\{
  \begin{matrix*}[l]
  {\colAll{\ccond{\cvar{b==1}}{\cvar{0}}{\cexpr{i}}}} &
  {\mathit{if} ~ \lnot\labelarith{\PubVars}{\cexpr{e}}}
  \\
  {\cexpr{i}} &
  {\mathit{otherwise}}
  \end{matrix*}
  \right.
  }
\end{align*}
}
\subcaptionbox{Flexible \iSLH \label{fig:faslh-instance}}[0.42\linewidth]{
\begin{align*}
 &\colFlexnew{
  \left\{
  \begin{matrix*}[l]
  {\colU{\cvar{b==0 \&\& } \cexpr{be}}} &
  {\mathit{if} ~ \lnot\labelbool{\PubVars}{\cexpr{be}}}
  \\
  {\colSel{\cexpr{be}}} &
  {\mathit{otherwise}}
  \end{matrix*}
  \right.
  }
\\
 &\colSel{
  \left\{
  \begin{matrix*}[l]
  {\colAll{\ccond{\cvar{b==1}}{\cvar{0}}{\cexpr{i}}}} &
  {\mathit{if} ~ \PubVars(\cvar{X}) \mathrel{\colFlexnew{\vee}} \colFlexnew{\lnot\labelarith{\PubVars}{\cexpr{i}}}}
  \\
  {\cexpr{i}} &
  {\mathit{otherwise}}
  \end{matrix*}
  \right.
  }
\\
 &\colSel{
  \left\{
  \begin{matrix*}[l]
  {\colAll{\ccond{\cvar{b==1}}{\cvar{0}}{\cexpr{i}}}} &
  {\mathit{if} ~ \lnot\labelarith{\PubVars}{\cexpr{e}} \mathrel{\colFlexnew{\vee}} \colFlexnew{\lnot\labelarith{\PubVars}{\cexpr{i}}}}
  \\
  {\cexpr{i}} &
  {\mathit{otherwise}}
  \end{matrix*}
  \right.
  }
\end{align*}
}
\subcaptionbox{\UltimateSLH \label{fig:uslh-instance}}[0.15\linewidth]{
\begin{align*}
 &\phantom{\left\{ \right.}
  \begin{matrix*}[l]
  \colU{\cvar{b==0 \&\& } \cexpr{be}}
  \\
  \phantom{\labelbool{\PubVars}{\cexpr{be}} : \LabelPublic}
  \end{matrix*}
\\
 &\phantom{\left\{ \right.}
  \begin{matrix*}[l]
  \ccond{\cvar{b==1}}{\cvar{0}}{\cexpr{i}}
  \\
  \phantom{\PubVars(\cvar{X}) : \LabelPublic}
  \end{matrix*}
\\
 &\phantom{\left\{ \right.}
  \begin{matrix*}[l]
  \ccond{\cvar{b==1}}{\cvar{0}}{\cexpr{i}}
  \\
  \phantom{\labelarith{\PubVars}{\cexpr{e}}}
  \end{matrix*}
\end{align*}
}
\caption{Overview of \iSLH variants instantiating recipe from \autoref{fig:aslh-template}}
\label{fig:aslh-variants}
\end{figure*}

Applying \iSLH to \autoref{ex:v1std} we obtain the following code:

\begin{listing}[Spectre v1 gadget protected with \iSLH]
\label{ex:v1hardened}

\begin{align*}
  &\cifs{\cvar{i < a$_1$\_size}}{\\
  &\quad\casgn{\cvar{b}}{\ccond{(\cvar{i < a$_1$\_size})}{\cvar{b}}{\cvar{1}}}\texttt{;}\\
  &\quad\caread{\cvar{j}}{\cvar{a}_1}{\ccond{\cvar{b==1}}{\cvar{0}}{\cvar{i}}}\texttt{;}\\
  &\quad\caread{\cvar{x}}{\cvar{a}_2}{\ccond{\cvar{b==1}}{\cvar{0}}{\cvar{j}}}\\&
  }{\\
    &\quad\casgn{\cvar{b}}{\ccond{\cvar{(i < a$_1$\_size)}}{\cvar{1}}{\cvar{b}}}&
  }
\end{align*}

\end{listing}

A Spectre v1 attacker can \emph{still} force the execution
of the \texttt{then} branch when $\cvar{i}$ is out of bounds, but this
will be detected and $\cvar{b}$ will be updated to reflect
it. Then the first array access will become $\cvar{a}_1\cvar{[0]}$ and will no
longer be out of bounds, and the attacker will be unable to exploit it to load secrets.

\subsection{Cryptographic Code Security}
\label{sec:specct}

Cryptographic implementations are common targets of timing side-channel attacks
aimed at extracting keys, even remotely over the Internet~\cite{BrumleyB05}.
To mitigate these dangers, cryptographic
engineers have developed secure programming disciplines, notably
cryptographic constant time (\CCT), now considered standard practice in the
field~\cite{AlmeidaBBDE16, CauligiSJBWRGBJ19, DanielBR23, BartheBCLP20}.
The \CCT discipline protects secrets from timing side-channel attacks
arising from secret-dependent branches and memory accesses.
For this, it imposes two requirements on the programmer:

\begin{enumerate}

\item
all program inputs must be identified as public or secret,
\item
control flow and memory addresses
that depend on secret inputs are not allowed.\footnote{
  Also variable-time arithmetic operations are not allowed on secrets,
  but these are not relevant for the simple setting of this paper.}

\end{enumerate}


Researchers have designed security analyses ensuring that these \CCT conditions
are met for real-world code~\cite{AlmeidaBBDE16, DanielBR23, BartheBCLP20}.
%
Even if the attacker can directly observe the branches that
are taken and the memory addresses that are accessed,
a \CCT program does not leak any secret information to the attacker. Prior
work has shown that this implies security against timing side-channel attacks
exploiting data and instruction caches~\cite{BartheBCLP20}.

\jb{The definition feels a little disconnected from the previous text.}

\begin{definition}[Cryptographic Constant-Time Security]
\label{def:cct-sec}
A program $c$ satisfies {\em \CCT security} w.r.t. public variables $P$
if for any public-equivalent states ($s_1 \sim_P s_2$), when
the program executes to completion on both states (\IE a small-step
semantics $\seqmulti{}$ fully evaluates the program from each state),
then both runs produce the same sequence of branching decisions and
memory accesses observable by an attacker ($\Obss_1 = \Obss_2$).
\[
\begin{array}{ll}
 s_1 \sim_P s_2 & \mathrel{\land}
  \langle c, s_1 \rangle \seqmulti{\Obss_1} \langle \cskip, \cdot \rangle\\
  & \mathrel{\land} \langle c, s_2 \rangle \seqmulti{\Obss_2} \langle \cskip, \cdot \rangle
  \Rightarrow \Obss_1 = \Obss_2
\end{array}
\]
%
\end{definition}


In the speculative world, an attacker can invalidate these
guarantees. By influencing the microarchitectural state of the
hardware, it can force a program to speculatively execute control
paths that are impossible in the standard, sequential semantics, and
leak secret data through the timing side channels CCT is supposed to
close~\cite{CauligiDGTSRB20,BartheCGKLOPRS21,VassenaDGCKJTS21,ShivakumarBBCCGOSSY23}.
To reason about these threats, we need a more informative semantics
that abstracts the speculative capabilities of an
attacker~\cite{ShivakumarBBCCGOSSY23}.
This \emph{speculative semantics} ($\specmulti{}{}$ below) adds rules that steer
execution along wrong branches and out-of-bounds memory accesses,
triggered by \emph{directives} $\Dirs$ issued by the attacker.
The speculative semantics also adds a boolean flag $b$ to the state to indicate
whether the processor is misspeculating.
We define speculative constant-time (SCT) security below in terms of
a notion of speculative observational equivalence:
%
%
\begin{definition}[Speculative \ObsEq]
\label{def-specobs}
Two speculative configurations $\langle c_1, s_1, b_1 \rangle$
and $\langle c_2, s_2, b_2 \rangle$ are observationally equivalent
w.r.t. the speculative semantics, written
$\langle c_1, s_1, b_1 \rangle \specobseq \langle c_2, s_2, b_2\rangle $ iff
\[
  \forall \Dirs \Obss_1 \Obss_2.~
    \langle c_1, s_1, b_1 \rangle \specmulti{\Obss_1}{\Dirs} \cdot \mathrel{\land}
    \langle c_2, s_2, b_2\rangle \specmulti{\Obss_2}{\Dirs} \cdot \Rightarrow
    \Obss_1 = \Obss_2
\]

\end{definition}

Compared to SCT definitions from the literature~\cite{CauligiDGTSRB20,
  DanielBR21,VassenaDGCKJTS21}, we are looking at all execution prefixes that do
not necessarily finish in a final state (the $\cskip$ in \autoref{def:cct-sec}).
This choice gives us stronger guarantees in the definitions below, \newreviews{in
  particular about the finite prefixes of nonterminating and stuck executions
  (which were ignored in \autoref{def:cct-sec})}.
%
We can directly compare observations $\Obss_1$ and $\Obss_2$ for equality here,
since the two speculative executions take identical directions $\Dirs$, which in
our setting implies that $|\Obss_1| = |\Dirs| = |\Obss_2|$.

\begin{definition}[Speculative Constant-Time Security]
\label{def:sct-sec}
%
A program $c$ satisfies SCT security if for any public-equivalent initial states
$s_1 \sim_P s_2$ and starting before misspeculation ($\BoolFalse$)
\[\langle c, s_1, \BoolFalse \rangle \specobseq \langle c, s_2, \BoolFalse \rangle\]

\end{definition}


\begin{example}
It is easy to see that \autoref{ex:v1std} is \CCT secure if the variables
are initially public. Even so, speculatively it leaks
secrets, as discussed in \autoref{sec:spectre}.
Formally, this program is not SCT secure, since for an out-of-bounds $\cvar{i}$,
the attacker can use a first directive in $\Dirs$ to force the program to
misspeculate on the \texttt{then} branch, and then a second directive to ask for
a secret in memory instead of the out-of-bounds value of $\cvar{a}_1[\cvar{i}]$.
The secret is loaded in $\cvar{j}$, and then made visible to the attacker through
the address of the second load operation, which is observed to have 
different values in $\Obss_1$ and $\Obss_2$.
On the other hand, after \iSLH protection (\autoref{ex:v1hardened}), this program is SCT
secure. In fact, any CCT\footnote{In general, when we write ``CCT'' unqualified
  we always mean the static CCT discipline, {\em not} the weaker semantic ``CCT
  security'' of \autoref{def:cct-sec}.}
program hardened with \iSLH satisfies SCT security, and as
we will see in the next subsection, a selective variant of \iSLH
is enough to enforce this.
\end{example}




\section{Key Ideas}
\label{sec:key-ideas}

\subsection{Selective \iSLH}


\SelectiveSLH
is a more efficient variant of \SLH that ``only
masks values speculatively loaded into publicly-typed variables''
\cite{ShivakumarBBCCGOSSY23},
exploiting the security levels used by the \CCT discipline to minimize the
number of masking operations while protecting secrets from
speculative attackers.
The original proposal uses a CCT type system to only mask the {\em values of
  loads into public variables} (\IE the \SelvSLH variant we discuss later in
  \autoref{sec:FvSLH}).
Yet the same optimization is also applicable to \iSLH, leading
to \emph{Selective index \SLH (\SeliSLH)}, which 
we show in \autoref{fig:saslh-instance} and which,
as far as we know, has not been studied formally before.
The masking of loads in $\llbracket \cdot \rrbracket_{\mathit{rd}}^\cdot$
uses a public map $\PubVars$ of variables to their boolean security levels
(public $\LabelPublic$ or secret $\LabelSecret$)
and functions $\labelarith{\PubVars}{\cexpr{e}}$
and $\labelbool{\PubVars}{\cexpr{b}}$ computing the security levels of
arithmetic and boolean expressions.

One crucial observation we make in this paper is that for achieving security for
\SeliSLH, it does {\em not} suffice to only mask the indices of loads, as shown
by the following counterexample we discovered during the formal proofs:



\begin{listing}[Leakage through unprotected stores]
  \label{ex:aslh-vs-vslh}
\begin{align*}
  &\cifsnoelse{\cvar{i < secrets\_size}}{\\
  &\quad\cawrite{\cvar{secrets}}{\cvar{i}}{\cvar{key}}\texttt{;}\\
  &\quad\caread{\cvar{x}}{\cvar{a}}{\cvar{0}}\texttt{;}\\
    &\quad\cifsnoelse{\cvar{x}}{ ... }
  }
  \end{align*}
\end{listing}

Once again, for an invalid $\cvar{i}$, the attacker can force the
program to execute the \texttt{then} branch. The first
instruction stores a secret value, $\cvar{key}$, normally into a secret array
$\cvar{secrets}$, but with the out-of-bounds index this
actually stores the secret at position $\cvar{0}$ of a public array $\cvar{a}$.
A subsequent and
seemingly innocuous in-bounds load from $\cvar{a[0]}$, where the secret
$\cvar{key}$ was already exfiltrated, then leaks its value to the attacker
via a subsequent timing observation ($\cifsnoelse{\cvar{x}}{ ... }$).
Since the $\cvar{a[0]}$ load is from (the in-bounds) index $\cvar{0}$, even
though speculation is detected, the \iSLH masking of this load will be a no-op
and thus alone will not prevent this counterexample.

The solution is to define $\llbracket \cdot \rrbracket_{\mathit{wr}}$ in
\autoref{fig:saslh-instance} to mask the indices of array stores when writing
secret values.
This extra masking causes the store operation in the example above
to write the $\cvar{key}$ in position $\cvar{0}$ of
$\cvar{secrets}$, leaving $\cvar{a}$ unchanged and thus preventing the attack.
%
With this addition we were able to formally prove speculative constant-time security:

\begin{theorem}
\label{thm:saslh-ct}
\SeliSLH enforces SCT security for CCT programs.
\end{theorem}

\subsection{Defining Relative Security for \UltimateSLH}

Most programs do not follow the \CCT discipline though,
so lightweight SLH defenses like the two \iSLH variants
above are not enough to enforce security for arbitrary programs.


\begin{listing}[Leakage through sequentially unreachable code]\footnote{\ifanon\else
  This simple counterexample was shared with us by Gilles Barthe.\fi{} \newreviews{It
  is easy to construct variants of this counterexample where one cannot statically
  determine that the branch condition evaluates to false.}}
\label{ex:gilles-cex}
\[
  \cifsnoelse{\cvar{false}}{
  \cifsnoelse{\cvar{secret = 0}}{
  ...}}
\]
\end{listing}

The code in the outer \texttt{then} branch will never be
sequentially executed, but a speculative attacker can force the
execution of this code. The inner branch condition reveals
information about the $\cvar{secret}$ variable, and no amount of
masking of memory accesses can prevent this.
Since it branches on $\cvar{secret}$, this code does not respect the CCT discipline,
but this is the kind code we still want to protect from Spectre v1 attackers.
To achieve this,
all sources of speculative leakage observations must be
masked, including branch conditions,
not only the indices of loads and stores.

\UltimateSLH (\USLH)~\cite{PatrignaniG21, ZhangBCSY23}
is a \SLH variant that offers protection to
arbitrary programs by {\em exhaustively} masking these three sources of
leakage, as shown in \autoref{fig:uslh-instance}. Like
vanilla \iSLH, it masks all indices of memory loads and stores.
Additionally, it also protects branch conditions by masking them with
the misspeculation flag. This causes branches to default to the false
case during misspeculation, preventing them from leaking information
about the original branch conditions. This prevents the
leakage in \autoref{ex:gilles-cex}.

Yet what is a suitable formal security definition for transformations protecting
arbitrary programs, like \USLH?
Because the arbitrary source program does not follow the CCT discipline it could
leak some data sequentially, and the transformations we look at here do not try
to prevent these sequential leaks.
Instead, defenses like \USLH enforce that the hardened program does not
leak any {\em more} information speculatively than the source program leaks
sequentially.
That is, if a sequential attacker against the source program cannot tell apart
two input states, the transformation ensures that a speculative attacker
cannot tell them apart either.
To formalize this relative security notion, we first also define a notion
of observationally equivalent configurations in the sequential semantics:

\begin{definition}[Sequential \ObsEq]
\label{def:seq-obseq}
Two source configurations $\langle c_1, s_1 \rangle$
and $\langle c_2, s_2 \rangle$ are observationally equivalent
in the sequential semantics, written $\langle c_1, s_1 \rangle \approx \langle c_2, s_2 \rangle$, iff
\[
  \forall \Obss_1 \Obss_2.~
    \seqmultis{\langle c_1, s_1 \rangle}{\Obss_1}{\cdot} \mathrel{\land}
    \seqmultis{\langle c_2, s_2 \rangle}{\Obss_2}{\cdot} \Rightarrow
    \Obss_1 \lessgtr \Obss_2
\]

\end{definition}

Looking at all execution prefixes above can result in observation sequences of
different lengths, so we require that one sequence is a prefix of the other
($\Obss_1 \lessgtr \Obss_2$).

\begin{definition}[Relative Security]
\label{def:rs-sec}
A program transformation $\transl{\cdot}$ satisfies relative security
if for all programs $c$ run from two arbitrary initial states $s_1$ and $s_2$ we have
\begin{align*}
\langle c, s_1 \rangle \seqobseq \langle c, s_2 \rangle \Rightarrow
& \langle \transl{c}, s_1, \BoolFalse \rangle \specobseq \langle \transl{c}, s_2, \BoolFalse \rangle
\end{align*}
\end{definition}

Looking at all execution prefixes in the sequential observational equivalence
premise of this definition is not a choice, but is forced
on us by the fact that a program that sequentially loops forever or ends with an
error (\IE gets stuck before reaching $\cskip$) can be forced by attacker
directions to successfully terminate speculatively.
So if we were to only look at successfully terminating sequential executions
(like in \autoref{def:cct-sec}) in the premise, we would simply not have any
information about the sequentially nonterminating or erroneous executions, and
we would thus not be able to prove (even a big-step version of)
the conclusion when such programs only successfully terminate speculatively.




Beyond these details,
the key idea is to relate the security of the
transformed program executed speculatively to the security of the {\em source}
program executed sequentially. This provides stronger guarantees than
previous definitions~\cite{PatrignaniG21, ZhangBCSY23}, which as discussed in
\autoref{sec:related-work}, \newreviews{deem secure a transformation
that introduces sequential leaks, instead of removing speculative ones}.
\newreviews{Our definition overcomes these limitations} and suitably captures
the security guarantees offered by \USLH, and later \FlexSLH.



\begin{theorem}
\label{thm:uslh-rs}
\USLH enforces relative security for all programs.
\end{theorem}


\subsection{The Best of Both Worlds: Flexible \iSLH}
\label{sec:FiSLH}

\SeliSLH and \USLH can be seen as falling on two
rather different points on the spectrum of Spectre v1 protections:

\begin{itemize}[leftmargin=*,nosep,label=$\blacktriangleright$]
\item \SeliSLH exploits developer information identifying which inputs
  are secret and which ones are not to implement an efficient countermeasure
  that guarantees no leakage of secret inputs whatsoever
  under the very strong requirement that the source program follows the \CCT discipline.

\item \USLH pays a much higher efficiency price to protect arbitrary programs
  and obtain a stronger relative security guarantee that includes programs that
  are not \CCT.
\end{itemize}

In this paper we propose \emph{Flexible index SLH (\FlexiSLH)}, a new hybrid design
shown in \autoref{fig:faslh-instance} that combines the desirable
features of both \SeliSLH and \USLH:

\begin{itemize}[leftmargin=*,nosep,label=$\blacktriangleright$]

\item Like \SeliSLH, \FlexiSLH also exploits information from
  the developer identifying which program inputs are secret and which
  are not, and only offers speculative protections to data
  that could depend on secret inputs.
  Unlike \SeliSLH though, \FlexiSLH does not require the strict CCT discipline.

\item Like \USLH, \FlexiSLH aims to offer relative security to arbitrary
  programs, which in our simplest formalization are only well-typed with respect
  to a standard IFC type system tracking explicit and
  implicit flows~\cite{VolpanoIS96} (see
  \autoref{sec:type-system}).\footnote{We lift this limitation in
      \autoref{sec:fsfvslh-key-ideas} and \autoref{sec:acom-and-analysis},
      by replacing this IFC type system with a
      flow-sensitive IFC analysis accepting all programs~\cite{HuntS06}.}
  Unlike \USLH though, \FlexiSLH uses the information about which values may
  have been influenced by secret inputs and which ones not to only selectively
  apply SLH protections, as detailed below.
%
%
%

\end{itemize}


As shown in \autoref{fig:faslh-instance}, the treatment of branch
conditions is a hybrid between \USLH (whose unique contributions are
given in \colU{red}) and \SeliSLH (in \colSel{blue}), with some new
logic (in \colFlexnew{purple}).
When the branch condition is a secret expression, it behaves like \USLH, using
the misspeculation flag to mask the branch condition.
When the branch condition is public, it behaves like \SeliSLH, and leaves the
branch condition unchanged.

For array loads, while \USLH always masks the indices, \SeliSLH only masks
indices for loads that sequentially involve public data.
This is, however, not enough for relative security of non-\CCT programs: since
speculative loads leak the index, \FlexiSLH also has to mask secret indices
to prevent the following speculative leak (similar to \autoref{ex:gilles-cex} above).

\begin{listing}[Leakage through sequentially unreachable load]
\label{ex:gilles-cex-load}
\[
  \cifsnoelse{\cvar{false}}{
  \caread{\cvar{xsecret}}{\cvar{a}}{\cvar{isecret}}}
\]
\end{listing}

This is similar for array stores, while \USLH masks indices for all stores,
\SeliSLH only masks indices for stores if the expression being stored is secret.
This is again not enough for relative security of non-\CCT programs: since
stores leak the index, \FlexiSLH also masks secret indices of stores.

\begin{listing}[Leakage through sequentially unreachable store]
\label{ex:gilles-cex-store}
\[
  \cifsnoelse{\cvar{false}}{
  \cawrite{\cvar{a}}{\cvar{isecret}}{\cvar{epublic}}}
\]
\end{listing}

We have proved in Rocq that the masking done by \FlexiSLH is enough to enforce
the following variant of \autoref{def:rs-sec}:

\begin{theorem}
\label{thm:faslh-rs}
\FlexiSLH enforces relative security for all IFC-well-typed programs and for all
public-equivalent initial states.
\end{theorem}

The extra premise that the initial states are public-equivalent intuitively
captures that public inputs are allowed to leak, as we do not consider differences
in observations resulting from different public inputs.
%
Perhaps a bit counterintuitively, because of the semantic equivalence premise
from \autoref{def:rs-sec}, an input labeled as secret is protected speculatively
by \FlexiSLH only if it does not leak sequentially.
This is the most one can achieve for non-\CCT programs, without heavier
defenses enforcing \CCT~\cite{CauligiSJBWRGBJ19}.


Finally, we proved in Rocq that \FlexiSLH is a generalization of \SeliSLH
and \USLH. The former means that CCT code pays no extra performance penalty
under \FlexiSLH compared to \SeliSLH:

\begin{theorem}[Connection between \FlexiSLH and \SeliSLH]
\label{thm:faslh-aslh}
For any \CCT program $c$ and labeling $\PubVars$ assigning security levels to variables, we have
\[
\translAux{\cexpr{c}}{\FlexiSLH}{\PubVars} = \translAux{\cexpr{c}}{\SeliSLH}{\PubVars}
\]
\end{theorem}

\begin{proof}[Proof intuition]
By inspection of \autoref{fig:saslh-instance}
and \autoref{fig:faslh-instance}.
\end{proof}

On the other hand, if all data is labeled as secret, then the program receives
exactly the same protection as with USLH:

\begin{theorem}[Connection between \FlexiSLH and \USLH]
\label{thm:faslh-uslh}
If all variables (both scalars and arrays) are labeled as
secret (\IE we use the labeling $(\lambda \_. \LabelSecret)$), 
then for any program $c$ we have
\[
\translAux{\cexpr{c}}{\FlexiSLH}{(\lambda \_. \LabelSecret)} 
= \translAux{\cexpr{c}}{\USLH}{}
\]
\end{theorem}

\begin{proof}[Proof intuition]
By inspection of \autoref{fig:faslh-instance}
and \autoref{fig:uslh-instance}.
\end{proof}







\subsection{\FsFlexvSLH: Flexible vSLH for All Programs}
\label{sec:fsfvslh-key-ideas}

For simplicity, our FiSLH variant above, and also our simple FvSLH variant of
\autoref{sec:FvSLH}, make use of an IFC type system~\cite{VolpanoIS96}
with a fixed assignment of variables as public or secret.
We show that this is not a conceptual limitation by devising \FsFlexvSLH, which
instead uses a flow-sensitive IFC analysis that accepts arbitrary source programs~\cite{HuntS06}
(\autoref{fig:flow-tracking} in \autoref{sec:acom-and-analysis}).

The IFC analysis statically computes whether the secrets in the
initial state can influence each variable used during the execution or not,
which results in a new variable assignment for the final state and an
annotated program in which program expressions like branch conditions are given
IFC levels.
The analysis is a sound static over-approximation of the semantics, joining the
IFC information computed over the two branches of conditionals and computing a
fixed-point for while loops.


Once the analysis annotates expressions with IFC
levels, \FsFlexvSLH uses these annotations to make masking decisions, but
otherwise it works the same as our FvSLH variant of \autoref{sec:FvSLH}.

We prove in Rocq that \FsFlexvSLH enforces relative security {\em for all}
source programs, the same as for Ultimate SLH:
\begin{theorem}
\label{thm:fsfvslh-rs}
\FsFlexvSLH enforces relative security for all source programs
and for all public-equivalent initial states.
\end{theorem}

Proving \FsFlexvSLH relatively secure is more challenging than for
FiSLH and FvSLH. All three proofs are inspired by
\citet{ShivakumarBBCCGOSSY23} and decompose the effort into a compiler
correctness proof with respect to an ideal semantics and a separate relative
security proof of the ideal semantics.
The ideal semantics we devised for \FsFlexvSLH is more interesting though, as it is
instrumented to dynamically and soundly track IFC on each execution step~\cite{SabelfeldR09}.
To show the relative security of the ideal semantics, we further introduce
a well-labeledness relation for annotated commands.
We prove that our static IFC analysis produces well-labeled commands and
that our dynamically-IFC-tracking ideal semantics preserves well-labeledness.
This interesting proof is detailed in \autoref{sec:fsfvslh}.

\section{Definitions}
\label{sec:defs}

\subsection{Language Syntax and Sequential Semantics}
\label{sec:defs-syntax-seq-semantics}

Throughout the paper, we use a simple while language
with arrays we call \SourceLang. Its syntax is shown
in \autoref{fig:syntax}. A first layer comprises pure arithmetic
expressions (\texttt{aexp}) on natural numbers (which we also write
$\cexpr{ae}$, or $\cexpr{ie}$ if they compute array indices), and
boolean expressions (\texttt{bexp}) (also written $\cexpr{be}$). The
second layer of commands (\texttt{com}) includes standard imperative
constructs, as well as array accesses.

\begin{figure}
  \centering
  \[
  \begin{aligned}
    \cexpr{e} \in\texttt{aexp} ::=& ~ \cexpr{n}\in\mathbb{N} && \text{number}\\
    |& ~ \texttt{X} \in \mathcal{V}&&\text{scalar variable}\\
    |& ~ \mathrm{op}_{\mathbb{N}}(\cexpr{e}, \ldots, \cexpr{e})&&\text{arithmetic operator}\\
    |& ~ \ccond{\cexpr{b}}{\cexpr{e}}{\cexpr{e}} &&\text{constant-time conditional}\\
    \cexpr{b}\in\texttt{bexp} ::=& ~ \BoolTrue ~|~\BoolFalse && \text{boolean}\\
    |&~ \mathrm{cmp}(\cexpr{e}, \cexpr{e})&&\text{arithmetic comparison}\\
    |&~ \mathrm{op}_{\mathbb{B}}(\cexpr{b},\ldots, \cexpr{b})&&\text{boolean operator}\\
    \cexpr{c}\in\texttt{com} ::=&~ \cskip &&\text{do nothing}\\
    |&~\casgn{\cvar{X}}{\cexpr{e}}&&\text{assignment}\\
    |&~\cseq{\cexpr{c}}{\cexpr{c}}&&\text{sequence}\\
    |&~\cifs{\cexpr{b}}{\cexpr{c}}{\cexpr{c}}&&\text{conditional}\\
    |&~\cwhiles{\cexpr{b}}{\cexpr{c}}&&\text{loop}\\
    |&~\caread{\cvar{X}}{\cvar{a}}{\cexpr{e}}&&\text{read from array}\\
    |&~\cawrite{\cvar{a}}{\cexpr{e}}{\cexpr{e}}&&\text{write to array}
  \end{aligned}
  \]
  \vspace{-1em}
  \caption{Language syntax}
  \label{fig:syntax}
\end{figure}

Program states are divided into scalar variables
$\cvar{X}, \cvar{Y}, \ldots \in \mathcal{V}$ and arrays
$\cvar{a}, \cvar{b}, \ldots \in \mathcal{A}$.
The scalar state $\ScalarState$ assigns
values to scalar variables, and is used to evaluate
arithmetic and boolean expressions involving scalar variables in the
usual way, written $\llbracket \cdot \rrbracket_{\rho}$. To update the
value of a variable, we write
$\subst{\cvar{X}}{v}{\rho}$.
The second half of the program state is the array state $\ArrayState$,
which assigns a fixed size $| \cvar{a} |_{\mu}$ to each
array, and defines a lookup function $\eval{\cvar{a}[i]}{\mu}$ to
fetch the value of an array at a given valid index, \IE in the range
$[0, | \cvar{a} |_{\mu})$. To update the value of an array at a valid
index, we write $\subst{\cvar{a}[i]}{v}{\mu}$.
The contents of an array cannot be used directly in arithmetic or
boolean expressions, and can only be transferred to and from scalar
variables through read and write commands.


We define a standard small-step operational semantics for \SourceLang
(for details, see \autoref{sec:spec-semantics}%
\ifappendix and also \autoref{sec:appendix-sequential-semantics}\fi).
States are triples
composed of the next command to execute, and the scalar and array
states. We write
$\seqstate{\cexpr{c}}{\rho}{\mu} \seqstep{o} \seqstate{\cexpr{c'}}{\rho'}{\mu'}$,
or $\seqstate{\cexpr{c}}{\rho}{\mu} \seqstep{o} \cdot$ if we do not
care about the state after the step.
%
A step produces an optional \emph{observation}
$o \in \mathit{Option}(\mathit{Obs})$ that represents 
information gained by an attacker via side channels.
We use the standard model for sequential side-channel
attackers, which observe the control flow of a program
($\branch{b}$ for a branch condition that evaluates to a boolean $b$)
and the accessed memory addresses ($\OARead{\cvar{a}}{i}$ and 
$\OAWrite{\cvar{a}}{i}$ resp. for array reads and writes at
$\cvar{a[i]}$). If a step does not produce an observation, we write
$\bullet$ instead of a concrete event.
%
%
Multi-step execution is the reflexive and transitive closure of the
step relation, written $\seqmulti{\Obss}$, where
$\Obss \in \mathit{List}(\mathit{Obs})$ is a trace of events.
In our formalization, silent steps leave no mark in the trace.
Full executions are multi-step executions of the form
$\seqstate{\cexpr{c}}{\rho}{\mu} \seqmulti{\Obss} \seqstate{\cskip}{\rho'}{\mu'}$.

Evaluating arithmetic and boolean
expressions does not produce any observable event, even though there
is a strong resemblance between the conditional expression
$\ccond{\cexpr{be}}{\cexpr{e_1}}{\cexpr{e_2}}$ and the conditional
command $\cifs{\cexpr{be}}{\cexpr{c_1}}{\cexpr{c_2}}$.
Internally, this imposes the use of branchless and unpredicted logic to implement
$\ccond{\cexpr{be}}{\cexpr{e_1}}{\cexpr{e_2}}$, which is critical to the
correctness of our \SLH countermeasures.


\begin{example}
\label{ex:v1std-formal}

Let us revisit \autoref{ex:v1std} using this formal model. There are
only two ways to fully evaluate that code under the sequential
semantics. Without loss of generality, suppose $\mu$ contains the
4-element array $\cvar{a}_1 = [0; 7; 1; 2]$ and another large array
$\cvar{a}_2$, say of size 1000. Suppose also that
$\lookup{\rho}{\cvar{a$_1$\_size}} = 4$.
%
Then, depending on the prospective index on $\cvar{a}_1$ given by
$\cvar{i}$:



\begin{enumerate}
\item
If the access is valid, \EG $\lookup{\rho}{\cvar{i}} = 1$, the
conditional produces an observation $\branch{\BoolTrue}$, and
afterwards the \texttt{then} branch executes, emitting two 
observations:
$\OARead{\cvar{a}_1}{1}$ and $\OARead{\cvar{a}_2}{7}$ (the
largest element in $\cvar{a}_1$ is 7, so all accesses on $\cvar{a}_2$
based on the contents of $\cvar{a}_1$ will be valid).
\item
If the access is not valid, \EG $\lookup{\rho}{\cvar{i}} = 4$, the
conditional produces an observation $\branch{\BoolFalse}$, and the
program terminates without any additional observations.
\end{enumerate}

We can see that under these conditions, the code is sequentially
secure, \IE out of bounds accesses cannot occur.
If we have a pair of states $\ScalarState_1, \ScalarState_2$ that
agree on the attacker-supplied value of $\cvar{i}$ (and the program
data, as described above), then the premise
of \autoref{def:rs-sec} is also satisfied.
\end{example}

\subsection{Speculative Semantics}
\label{sec:spec-semantics}

The observations produced by executions in the sequential
semantics allow us to reason about leakage that follows necessarily
from the intended control flow and memory access patterns of the
program. This is enough for properties like CCT in the standard
leakage model, where the attacker is passive.
%
However, a speculative attacker is able to actively influence the
program by causing it to veer outside the paths prescribed by the
sequential semantics,
and these transient paths empower the attacker
to make additional observations about the state of the program.
To model these more powerful attackers, we use a speculative
semantics \cite{ShivakumarBBCCGOSSY23} that reflects these active
capabilities, shown in \autoref{fig:spec-semantics}. (This is a strict
extension of the sequential semantics, which we recover by removing
the additions highlighted in \colHighlight{red}.)


\begin{figure}
  \centering
  \[
  \inferrule*[left=Spec\_Asgn]
  {v = \eval{\cexpr{ae}}{\rho}}
  {\SpecEval{\casgn{\cvar{X}}{\cexpr{ae}}}{\cexpr{\rho}}{\cexpr{\mu}}{\highlight{b}}{\cskip}
   {\subst{\cvar{X}}{v}{\rho}}
   {\cexpr{\mu}}{\highlight{b}}{\bullet}{\highlight{\bullet}}}
  \]\[
  \inferrule*[left=Spec\_Seq\_Step]
  {\SpecEval{\cexpr{c_1}}{\rho}{\mu}{\highlight{b}}{\cexpr{c_1'}}{\rho'}{\mu'}{\highlight{b'}}{o}{\highlight{d}}}
  {\SpecEval{\cseq{\cexpr{c_1}}{\cexpr{c_2}}}{\rho}{\mu}{\highlight{b}}{\cseq{\cexpr{c_1'}}{\cexpr{c_2}}}{\rho'}{\mu'}{\highlight{b'}}{o}{\highlight{d}}}
  \]
  \[
  \infer[Spec\_While]
  {\cexpr{c_{while}} = \cwhiles{\cexpr{be}}{\cexpr{c}}}
  {\SpecEval{\cexpr{c_{while}}}{\rho}{\mu}{\highlight{b}}{\cifs{\cexpr{be}}{\cseq{\cexpr{c}}{\cexpr{c_{while}}}}{\cskip}}{\rho}{\mu}{\highlight{b}}{\bullet}{\highlight{\bullet}}}
  \]\[
  \inferrule*[left=Spec\_Seq\_Skip]
  {\phantom{\quad}}
  {\SpecEval{\cseq{\cskip}{\cexpr{c}}}{\rho}{\mu}{\highlight{b}}{\cexpr{c}}{\rho}{\mu}{\highlight{b}}{\bullet}{\highlight{\bullet}}}
  \]
  \[
  \infer[Spec\_If]
  {b' = \eval{\cexpr{be}}{\rho}}
  {\SpecEval{\cifs{\cexpr{be}}{\cexpr{c_\BoolTrue}}{\cexpr{c_\BoolFalse}}}{\rho}{\mu}{\highlight{b}}{\cexpr{c}_{b'}}{\rho}{\mu}{\highlight{b}}{\branch{b'}}{\highlight{\step}}}
  \]\[
\highlight{
  \infer[Spec\_If\_Force]
  {b' = \eval{\cexpr{be}}{\rho}}
  {\SpecEval{\cifs{\cexpr{be}}{\cexpr{c_\BoolTrue}}{\cexpr{c_\BoolFalse}}}{\rho}{\mu}{b}{\cexpr{c}_{\neg b'}}{\rho}{\mu}{\BoolTrue}{\branch{b'}}{\force}}
}
  \]
  \[
  \infer[Spec\_Read]
  {i = \eval{\cexpr{ie}}{\rho} \\
   v = \eval{\cvar{a}[i]}{\mu} \\
   i < | \cvar{a} |_{\mu}
  }
  {\SpecEval{\caread{\cvar{X}}{\cvar{a}}{\cexpr{ie}}}{\rho}{\mu}{\highlight{b}}{\cskip}
   {\subst{\cvar{X}}{v}{\rho}}
   {\mu}{\highlight{b}}{\OARead{\cvar{a}}{i}}{\highlight{\step}}}
  \]\[
\highlight{
  \infer[Spec\_Read\_Force]
  {i = \eval{\cexpr{ie}}{\rho} \\
   v = \eval{\cvar{b}[j]}{\mu} \\
   i \ge | \cvar{a} |_{\mu} \\
   j < | \cvar{b} |_{\mu}
  }
  {\SpecEval{\caread{\cvar{X}}{\cvar{a}}{\cexpr{ie}}}{\rho}{\mu}{\BoolTrue}{\cskip}
   {\subst{\cvar{X}}{v}{\rho}}
   {\mu}{\BoolTrue}{\OARead{\cvar{a}}{i}}{\DLoad{\cvar{b}}{j}}}
}
  \]
  \[
  \infer[Spec\_Write]
  {i = \eval{\cexpr{ie}}{\rho} \\
   v = \eval{\cexpr{ae}}{\rho} \\
   i < | \cvar{a} |_{\mu}
  }
  {\SpecEval{\cawrite{\cvar{a}}{\cexpr{ie}}{\cexpr{ae}}}{\rho}{\mu}{\highlight{b}}{\cskip}{\rho}
   {\subst{\cvar{a}[i]}{v}{\mu}}
   {\highlight{b}}{\OAWrite{\cvar{a}}{i}}{\highlight{\step}}}
  \]\[
\highlight{
  \infer[Spec\_Write\_Force]
  {i = \eval{\cexpr{ie}}{\rho} \\
   v = \eval{\cexpr{ae}}{\rho} \\
   i \ge | \cvar{a} |_{\mu} \\
   j < | \cvar{b} |_{\mu}
  }
  {\SpecEval{\cawrite{\cvar{a}}{\cexpr{ie}}{\cexpr{ae}}}{\rho}{\mu}{\BoolTrue}{\cskip}{\rho}
   {\subst{\cvar{b}[j]}{v}{\mu}}
   {\BoolTrue}{\OAWrite{\cvar{a}}{i}}{\DStore{\cvar{b}}{j}}}
}
  \]
  \vspace{-1em}
  \caption{Speculative semantics}
  \label{fig:spec-semantics}
\end{figure}

The small-step relation
$\specstate{\cexpr{c}}{\rho}{\mu}{b} \specstep{o}{d} \specstate{\cexpr{c}'}{\rho'}{\mu'}{b'}$
is a generalization of the sequential relation, extended in two
ways. The first change is the addition of an optional speculation
directive $d \in \mathit{Option}(\mathit{Dir})$ that abstracts the
active capabilities of the attacker to steer execution down a
specific speculative path. Observation-producing rules in the
sequential semantics are split into two speculative variants. The first
variant behaves exactly like its sequential counterpart, and does so
when allowed to by the attacker using the directive
$\step$. The second variant misspeculates, executing a branch that
should not be taken (when issuing directive $\force$ for
conditionals), or performing an arbitrary array access (when issuing
directives $\DLoad{\cvar{a}}{i}$ and $\DStore{\cvar{a}}{i}$ resp. for
reads and writes).
Note the initial speculation flag in
\textsc{Spec\_Read\_Force} and \textsc{Spec\_Write\_Force}:
the attacker can only cause misspeculation on array accesses after having forced a
branch misprediction, and only if it manages to force an
out-of-bounds access. In this case, we conservatively give the
attacker choice over the resulting access,
overapproximating the speculative capabilities of a Spectre v1 attacker.
We write $\bullet$ for steps where no directive is needed.

The second change is the extension of program states with a fourth
component, a boolean flag $b$ that indicates whether the program has
diverged from its sequential execution. The flag is set to
$\BoolFalse$ at the start of the execution, and updated to $\BoolTrue$
when the attacker forces a branch misprediction.

Again, we write $\specmulti{\Obss}{\Dirs}$ for multi-step execution,
where $\Dirs$ is the sequence of speculative directives selected by
the attacker.
Because our security definitions quantify over all possible execution
paths, the speculative semantics does not need to model implementation
details like finite speculation windows or rollback
\cite{ShivakumarBBCCGOSSY23}, as also discussed in \autoref{sec:related-work}.



\begin{example}

  \label{ex:v1std-not-rs}
Using the sequential and speculative semantics for \SourceLang, we can
show that \autoref{ex:v1std} without protections is not relative secure,
\IE does not satisfy \autoref{def:rs-sec}.

Assume the scalar variables $\cvar{i}$ and $\cvar{a$_1$\_size}$ and
the arrays $\cvar{a}_1$ and $\cvar{a}_2$ are public, and all other
data is secret. As we have seen in \autoref{ex:v1std-formal}, the
program satisfies \autoref{def:cct-sec}, and therefore for any pair of
public-equivalent states also \autoref{def:seq-obseq}.
%
%
%
But this does not mean that two executions from related states with
the same directives will yield identical observations.
Let us return to the second scenario in the last
example, where $\lookup{\rho}{\cvar{i}} = 4$ in an attempt to access
$\cvar{a}_1$ out of bounds.
Suppose that $\mu_1$ and $\mu_2$ are identical except for the fact
that $\mu_1$ contains a 1-element array $\cvar{a}_3 = [42]$, but in
$\mu_2$ we have $\cvar{a}_3 = [43]$.
The attacker can falsify the property as follows:
\begin{enumerate}
\item
Issue a directive $\force$ to speculatively begin executing the
\texttt{then} branch, even though the condition is
false. This produces an observation $\branch{\BoolFalse}$ on
both runs.
\item
The first read is an out of bounds access, so the attacker can choose
a directive $\DLoad{\cvar{a}_3}{0}$. This produces an observation
$\OARead{\cvar{a}_1}{4}$ on both runs, but the scalar states now
differ: $\cvar{j}$ holds the secret 42 in $\rho_1$ and 43 in
$\rho_2$.
\item
The second read is allowed to proceed normally using 
$\step$, 
producing the observation
$\OARead{\cvar{a}_2}{42}$ on the first run and $\OARead{\cvar{a}_2}{43}$
on the second, revealing the secret.
\end{enumerate}
\end{example}
%
%
%
\subsection{Simple IFC Type System}
\label{sec:type-system}
\begin{figure*}
\centering
\[
\infer[WT\_Skip]
{\quad}
{\PubVars; \PubArrs \vdash_{\highlight{\mathpc}} \cskip}
\quad
\infer[WT\_Asgn]
{\labelarith{\PubVars}{\cexpr{a}} = \ell \quad
 \highlight{\mathpc \sqcup} \ell \sqsubseteq \PubVars(\cvar{X})}
{\PubVars; \PubArrs \vdash_{\highlight{\mathpc}} \casgn{\cvar{X}}{\cvar{a}}}
\quad
\infer[WT\_Seq]
{\PubVars; \PubArrs \vdash_{\highlight{\mathpc}} \cexpr{c_1} \quad
 \PubVars; \PubArrs \vdash_{\highlight{\mathpc}} \cexpr{c_2}}
{\PubVars; \PubArrs \vdash_{\highlight{\mathpc}} \cseq{\cexpr{c_1}}{\cexpr{c_2}}}
\quad
\infer[WT\_If]
{\labelbool{\PubVars}{\cexpr{be}} \highlight{ = \ell} \quad
 \PubVars; \PubArrs \vdash_{\highlight{\mathpc \sqcup \ell}} \cexpr{c_1} \quad
 \PubVars; \PubArrs \vdash_{\highlight{\mathpc \sqcup \ell}} \cexpr{c_2}}
{\PubVars; \PubArrs \vdash_{\highlight{\mathpc}} \cifs{\cexpr{be}}{\cexpr{c_1}}{\cexpr{c_2}}}
\]
\[
\infer[WT\_While]
{\labelbool{\PubVars}{\cexpr{be}} \highlight{ = \ell} \quad
 \PubVars; \PubArrs \vdash_{\highlight{\mathpc \sqcup \ell}} \cexpr{c}}
{\PubVars; \PubArrs \vdash_{\highlight{\mathpc}} \cwhiles{\cexpr{be}}{\cexpr{c}}}
\quad
\infer[WT\_ARead]
{\labelarith{\PubVars}{\cexpr{i}} \highlight{ = \ell_{\cexpr{i}}} \quad
 \highlight{\mathpc \sqcup \ell_{\cexpr{i}} \sqcup} \PubArrs(\cvar{a}) \sqsubseteq \PubVars(\cvar{X})}
{\PubVars; \PubArrs \vdash_{\highlight\mathpc} \caread{\cvar{X}}{\cvar{a}}{\cexpr{i}}}
\quad
\infer[WT\_AWrite]
{\labelarith{\PubVars}{\cexpr{i}} \highlight{ = \ell_{\cexpr{i}}} \quad
 \labelarith{\PubVars}{\cexpr{e}} = \ell \quad
 \highlight{\mathpc \sqcup \ell_{\cexpr{i}} \sqcup} \ell \sqsubseteq \PubArrs(\cvar{a})}
{\PubVars; \PubArrs \vdash_{\highlight{\mathpc}} \cawrite{\cvar{a}}{\cexpr{i}}{\cexpr{e}}}
\]
\vspace{-1em}
\caption{IFC type system used by \FlexiSLH and \FlexvSLH}
\label{fig:ifc-type-system}
\end{figure*}

As mentioned in \autoref{sec:FiSLH}, our current implementation of flexible index SLH
in Rocq makes use of a simple IFC type system \emph{à la}
Volpano-Smith~\cite{VolpanoIS96} 
that tracks explicit and implicit flows in
\SourceLang programs
(\autoref{fig:ifc-type-system}).
For this, we make use of
a pair of maps that assign security labels to
scalar variables ($\PubVars: \mathcal{V} \rightarrow \mathcal{L}$) and
to array variables ($\PubArrs: \mathcal{A} \rightarrow \mathcal{L}$),
respectively. We consider the two-point lattice of boolean levels,
with $\LabelPublic$ for public and $\LabelSecret$ for secret.
Since public may flow into secret, we take 
$\LabelPublic \sqsubseteq \LabelSecret$.
We lift the map of public variables $\PubVars$ to a pair of functions
$\labelbool{\PubVars}{\cexpr{be}}$ and
$\labelarith{\PubVars}{\cexpr{ae}}$ that compute the security levels of
arithmetic and boolean expressions in the usual way.
We define a standard public-equivalence relation
$\ScalarState_1 \sim_{\PubVars} \ScalarState_2$ that relates two scalar states
iff they agree on the public variables according to $\PubVars$, and
similarly $\ArrayState_1 \sim_{\PubArrs} \ArrayState_2$ iff they agree on the
sizes and contents of the public arrays according to $\PubArrs$.

We define the standard typing judgment
$\PubVars; \PubArrs \vdash_{\LabelPublic} \cexpr{c}$ using the rules from
\autoref{fig:ifc-type-system}.
This will ensure that command $\cexpr{c}$, run from two public-equivalent states
is secure, in the sense that secret variables do not affect public ones.
This type system generalizes the
type systems used for \CCT (\EG in \autoref{sec:specct});
\autoref{fig:ifc-type-system} highlights
the extensions of our type system w.r.t. the \CCT type
system in \colHighlight{red}: 
the labels of branch conditions and indices are no longer required to be public.
Instead, we maintain a PC label to track labels of surrounding branch conditions,
and adjust the read and write cases to account for implicit flows using the PC and index labels.

\rb{Show any preservation lemmas?}\ch{No space}


\section{Formal Results for \FlexiSLH}
\label{sec:formal-results}

The main goal of this section is to prove that \FlexiSLH enforces
relative security (\autoref{def:rs-sec}), and from there to derive
security of other \SLH variants that are special cases of \FlexiSLH.

\subsection{Ideal Semantics}
\label{sec:ideal-aslh}

Similarly to \citet{ShivakumarBBCCGOSSY23}, 
our security proofs use an auxiliary semantics which refines
the speculative semantics
from \autoref{fig:spec-semantics}. This \emph{ideal semantics}
introduces new restrictions, which reflect the idealized behavior of
programs that are hardened against speculative leaks by \FlexiSLH.


\begin{figure}
  \centering
  \[
  \infer[Ideal\_If]
  {\highlight{\labelbool{\PubVars}{\cexpr{be}} = \ell} \\
   b' = \highlight{(\ell \vee \neg b)} \mathrel{\highlight{\wedge}} \eval{\cexpr{be}}{\ScalarState}
  }
  {\IdealEval{\cifs{\cexpr{be}}{\cexpr{c_\BoolTrue}}{\cexpr{c_\BoolFalse}}}{\rho}{\mu}{b}{\cexpr{c_{b'}}}{\rho}{\mu}{b}{\branch{b'}}{\step}}
  \]\[
  \infer[Ideal\_If\_Force]
  {\highlight{\BoolLabel{\cexpr{be}} = \ell} \\
   b' = \highlight{(\ell \vee \neg b)} \mathrel{\highlight{\wedge}} \eval{\cexpr{be}}{\ScalarState}
  }
  {\IdealEval{\cifs{\cexpr{be}}{\cexpr{c_\BoolTrue}}{\cexpr{c_\BoolFalse}}}{\rho}{\mu}{b}{\cexpr{c_{\neg b'}}}{\rho}{\mu}{\BoolTrue}{\branch{b'}}{\force}}
  \]
  \[
  \infer[Ideal\_Read]
  {\highlight{\labelarith{\PubVars}{\cexpr{ie}} = \ell_{\cexpr{i}}} \\
   i = \highlight{
   \left\{
   {\begin{matrix*}[l]
    0 & \mathit{if} ~ (\lnot\ell_{\cexpr{i}} \vee \lookup{\PubVars}{\cvar{X}}) \wedge b \\
    \colAll{\eval{\cexpr{ie}}{\rho}} & \mathit{otherwise}
    \end{matrix*}
   }
   \right.
   }
   \\\\
   v = \eval{\cvar{a}[i]}{\mu} \\
   i < | \cvar{a} |_{\mu}
  }
  {\IdealEval{\caread{\cvar{X}}{\cvar{a}}{\cexpr{ie}}}{\rho}{\mu}{b}{\cskip}
   {\subst{\cvar{X}}{v}{\rho}}
   {\mu}{b}{\OARead{\cvar{a}}{i}}{\step}}
  \]\[
  \infer[Ideal\_Read\_Force]
  {\highlight{\labelarith{\PubVars}{\cexpr{ie}}} \\
   \highlight{\lnot\lookup{\PubVars}{\cvar{X}}} \\
   i = \eval{\cexpr{ie}}{\rho} \\\\
   v = \eval{\cvar{b}[j]}{\mu} \\
   i \ge | \cvar{a} |_{\mu} \\
   j < | \cvar{b} |_{\mu}
  }
  {\IdealEval{\caread{\cvar{X}}{\cvar{a}}{\cexpr{ie}}}{\rho}{\mu}{\BoolTrue}{\cskip}
   {\subst{\cvar{X}}{v}{\rho}}
   {\mu}{\BoolTrue}{\OARead{\cvar{a}}{i}}{\DLoad{\cvar{b}}{j}}}
  \]
  \[
  \infer[Ideal\_Write]
  {i = \highlight{
   \left\{
   {\begin{matrix*}[l]
    0 & \mathit{if} ~ (\lnot\ell_{\cexpr{i}} \vee \lnot\ell) \wedge b \\
    \colAll{\eval{\cexpr{ie}}{\rho}} & \mathit{otherwise}
    \end{matrix*}
   }
   \right.
   }
\\\\
   \highlight{\labelarith{\PubVars}{\cexpr{ie}} = \ell_{\cexpr{i}}} \\
   \highlight{\labelarith{\PubVars}{\cexpr{ae}} = \ell} \\
   v = \eval{\cexpr{ae}}{\rho} \\ 
   i < | \cvar{a} |_{\mu}
  }
  {\IdealEval{\cawrite{\cvar{a}}{\cexpr{ie}}{\cexpr{ae}}}{\rho}{\mu}{b}{\cskip}{\rho}
   {\subst{\cvar{a}[i]}{v}{\mu}}
   {b}{\OAWrite{\cvar{a}}{i}}{\step}}
  \]\[
  \infer[Ideal\_Write\_Force]
  {\highlight{\labelarith{\PubVars}{\cexpr{ie}}} \\
  \highlight{\labelarith{\PubVars}{\cexpr{ae}}} \\\\
   i = \eval{\cexpr{ie}}{\rho} \\
   v = \eval{\cexpr{ae}}{\rho} \\  
   i \ge | \cvar{a} |_{\mu} \\
   j < | \cvar{b} |_{\mu}
  }
  {\IdealEval{\cawrite{\cvar{a}}{\cexpr{ie}}{\cexpr{ae}}}{\rho}{\mu}{\BoolTrue}{\cskip}{\rho}
   {\subst{\cvar{b}[j]}{v}{\mu}}
   {\BoolTrue}{\OAWrite{\cvar{a}}{i}}{\DStore{\cvar{b}}{j}}}
  \]
  \vspace{-1em}
  \caption{Ideal semantics for \FlexiSLH (selected rules)}
  \label{fig:ideal-semantics-faslh}
\end{figure}

Again, we define the semantics as a small-step relation and write
$\PubVars \vdash \IdealEval{\cexpr{c}}{\rho}{\mu}{b}{\cexpr{c'}}{\rho'}{\mu'}{b'}{\Obss}{\Dirs}$.
The labeling map $\PubVars$ is constant throughout, so we
elide it and the turnstile from the relation, although we may refer
to it inside premises.
The speculative and the ideal semantics share the same collection of
rules. \autoref{fig:ideal-semantics-faslh} shows the subset of rules
that change from the speculative semantics, with additions highlighted
in \colHighlight{red}.

Only rules which generate observations are modified,
and all changes are exclusively additions to premises.
The two rules for conditionals simply mask the branch condition so
that it defaults to the \texttt{else} branch when the branch condition is secret and the speculation flag is set.
This prevents leaking information about such conditions, such as in
\autoref{ex:gilles-cex}.\jb{This is previously explained in 3.2 for USLH, can we reference this?}
Similarly, the $\step$ rules for array loads and stores add conditions
under which the index 
is masked, so that it will be $\cvar{0}$
when the speculation flag is set, and therefore guaranteed to be in bounds.
The misspeculating rules on array loads and stores, on the other hand, only apply
under certain safe scenarios, and get stuck when these are not met.
For example, a misspeculating read is only allowed to proceed if the index is public
and the result is stored to a secret variable. This stems from the fact that the operation
leaks the index, but also allows the attacker to read from any (public or secret) array,
as demonstrated in \autoref{ex:v1std-not-rs}.
Finally, multi-step executions $\idealmultiarrow{\Obss}{\Dirs}$ and ideal
observational equivalence $\approx_i$ are defined in precisely the
same way as their speculative counterparts.

\subsection{Key Theorems}

There are several key lemmas used in our proof
of \autoref{thm:faslh-rs}. The first of these is a backwards compiler
correctness (BCC) result which shows that the program hardened by \FlexiSLH executed with the
unrestricted speculative semantics behaves like the original source
program executed with the ideal semantics, \IE the \FlexiSLH translation
correctly implements the restrictions of the ideal semantics.
Despite their parallels, establishing the links between the two is
technically subtle. \rb{Do we want to mention loops here, from Léon's
comment?}


\begin{lemma}[Backwards compiler correctness]
\label{lem:bcc}
\setcounter{equation}{0}
\begin{align}
 &(\forall \cvar{a}. | \cvar{a} |_{\ArrayState} > 0) \wedge
  \cvar{b} \notin \UsedVars{\cexpr{c}} \wedge
  \lookup{\ScalarState}{\cvar{b}} = \natofbool{\SemSpecFlag}
 &\Rightarrow
 \label{eq:bcc-sideconds}
\\
 &\specstate{\transls{\cexpr{c}}{\FlexiSLH}}{\ScalarState}{\ArrayState}{\SemSpecFlag}
  \specmulti{\Obss}{\Dirs}
  \specstate{\cexpr{c'}}{\ScalarState'}{\ArrayState'}{\SemSpecFlag'}
 &\Rightarrow
 \label{eq:bcc-target}
\\
 &\exists c''.
    \idealmulti
    {\specstate{\cexpr{c}}{\ScalarState}{\ArrayState}{\SemSpecFlag}}
    {\Obss}{\Dirs}
    {\specstate{\cexpr{c''}}{\subst{\cvar{b}}{\lookup{\ScalarState}{\cvar{b}}}{\ScalarState'}}{\ArrayState'}{\SemSpecFlag'}}
 &\wedge
  \label{eq:bcc-source}
\\
 &\qquad
  (
    \cexpr{c'} = \cskip \Rightarrow
    \cexpr{c''} = \cskip \wedge
    \lookup{\ScalarState'}{\cvar{b}} = \natofbool{b'}
  )
  \label{eq:bcc-strenghtened}
\end{align}
\end{lemma}

\begin{proof}[Explanation and proof outline]

This lemma connects a ``target run'' of the hardened command 
\eqref{eq:bcc-target} to a ``source run'' of the original command
\eqref{eq:bcc-source}. The initial configurations of both runs are
identical except for the translation function; the memories and the
speculation flag in the final states are also identical. Additionally,
it concludes that terminating target
runs correspond to terminating source runs, and that for these, the
values of the speculation flags maintained by the program and the semantics
coincide \eqref{eq:bcc-strenghtened}.

The result holds under a number of side conditions \eqref{eq:bcc-sideconds}.
For one, we require all arrays to be non-empty, so that when
we mask an index to 0, this does not result in an out-of-bound access
that could be misspeculated upon.
For another, the original program must not use the variable reserved
for the speculation flag (hence why $\cvar{b}$) remains unchanged throughout
the source run), and last, the initial value of $\cvar{b}$ must be the same
as the semantic flag $b$. The function $\natofbool{\cdot}$ simply produces
the natural encoding of the input boolean.

The proof proceeds by induction on the number of steps in the target
run \eqref{eq:bcc-target}, followed by a case analysis on the
command. Note however that the number of steps cannot be inferred from
the size of either the observations or the directives due to the
presence of silent steps. Instead, we use induction on
$| \cexpr{c} | + | \Obss |$, where $| \cexpr{c} |$ is
defined simply as the number of its constructors. This compound
measure effectively bounds the number of steps taken by the
speculative semantics in the target run. Loops are handled by the size
of the trace of observations, as each iteration produces at least one
observation.
\end{proof}

A property of the ideal semantics used in the following lemmas
allows us to relate the results of a pair of runs from related
states when they are given the same directives and produce the same
observations.
Here we show this for single steps, but we have generalized this to multi-steps
%
by ensuring that $\cexpr{c_1}$ and $\cexpr{c_2}$
cannot reduce without producing an observation.

\begin{lemma}[Noninterference of $\idealsteparrow{}{}$ with equal observations]
  \label{lem:faslh-ideal-noninterference}
  \begin{align*}
    &\wtifc{\mathpc}{\cexpr{c}}\wedge \ScalarState_1 \sim_{\PubVars} \ScalarState_2\wedge
    \ArrayState_1 \sim_{\PubArrs} \ArrayState_2 &\Rightarrow\\
    &\IdealEval{\cexpr{c}}{\ScalarState_1}{\ArrayState_1}{b}{\cexpr{c_1}}{\ScalarState_1'}{\ArrayState_1'}{b_1}{o}{d}&\Rightarrow\\
    &\IdealEval{\cexpr{c}}{\ScalarState_2}{\ArrayState_2}{b}{\cexpr{c_2}}{\ScalarState_2'}{\ArrayState_2'}{b_2}{o}{d}&\Rightarrow\\
    &\cexpr{c_1} = \cexpr{c_2} \wedge b_1 = b_2 \wedge \ScalarState_1' \sim_{\PubVars} \ScalarState_2' \wedge \ArrayState_1' \sim_{\PubArrs} \ArrayState_2' &
  \end{align*}
\end{lemma}
\begin{proof}[Proof sketch]
  By induction on the derivation of the first step, and inversion on the second one. 
\end{proof}

Another key technical result relates the behaviors of well-typed
programs running from public-equivalent but misspeculating states, still
in the ideal semantics~\cite{ZhangBCSY23}.
%
\rb{Léon's comment, with some sharpening, could work better inside the proof proper:
``The idea is that since each branching operation containing a secret is masked, there is no way to obtain different observations after misspeculating based on publicly equivalent states.''
}


\begin{lemma}[Unwinding of ideal misspeculated executions]
\label{lem:gilles-lemma}
\setcounter{equation}{0}
\begin{align}
 &\wtifc{\LabelPublic}{\cexpr{c}} \wedge
  \ScalarState_1 \sim_{\PubVars} \ScalarState_2 \wedge
  \ArrayState_1 \sim_{\PubArrs} \ArrayState_2
 &\Rightarrow
 \label{eq:gilles-loweq}
\\
 &\specstate{\cexpr{c}}{\ScalarState_1}{\ArrayState_1}{\BoolTrue}
  \approx_i
  \specstate{\cexpr{c}}{\ScalarState_2}{\ArrayState_2}{\BoolTrue}
  \label{eq:gilles-goal}
\end{align}
\end{lemma}

\begin{proof}


By induction on one of the speculative executions
exposed after unfolding \eqref{eq:gilles-goal}.
The main auxiliary lemma is a single-step version of the same statement,
which additionally states that the two steps result in the same
command. This follows from the initial
hypothesis \eqref{eq:gilles-loweq}, \autoref{lem:faslh-ideal-noninterference},
and the fact that the ideal step relation preserves well-typedness.
%
%
\end{proof}

The final piece of the puzzle requires showing that the ideal
semantics satisfies relative security.


\begin{lemma}[$\idealmultiarrow{}{}$ ensures relative security]
\label{lem:faslh-ideal-rs}
\setcounter{equation}{0}
\begin{align}
 &\wtifc{\LabelPublic}{\cexpr{c}} \wedge
  \ScalarState_1 \sim_{\PubVars} \ScalarState_2 \wedge
  \ArrayState_1 \sim_{\PubArrs} \ArrayState_2
 &\Rightarrow
  \label{eq:faslh-ideal-loweq}
\\
 &\seqstate{\cexpr{c}}{\ScalarState_1}{\ArrayState_1}
  \approx
  \seqstate{\cexpr{c}}{\ScalarState_2}{\ArrayState_2}
 &\Rightarrow
  \label{eq:faslh-ideal-seqeq}
\\
 &\specstate{\cexpr{c}}{\ScalarState_1}{\ArrayState_1}{\BoolFalse}
  \approx_i
  \specstate{\cexpr{c}}{\ScalarState_2}{\ArrayState_2}{\BoolFalse}
  \label{eq:faslh-ideal-goal}
\end{align}
\end{lemma}

\begin{proof}

Unfolding the definition in \eqref{eq:faslh-ideal-goal} exposes the
two ideal executions together with their shared directives $\Dirs$ and
their observations $\Obss_1$ and $\Obss_2$. The goal of the theorem is
to establish their equality based on the equivalence of observations
in the sequential semantics \eqref{eq:faslh-ideal-seqeq},
well-typedness, and public-equivalence of the initial states
\eqref{eq:faslh-ideal-loweq}.

If the attacker never forces misspeculation, \IE $\Dirs$ contains
exclusively $\step$ directives, the goal follows directly because
ideal executions without misspeculation are identical to sequential executions
and thus preserve the property
in \eqref{eq:faslh-ideal-seqeq}.

If the attacker forces the program to misspeculate at some point
during the executions, the list of directives is necessarily of the
form $\Dirs = [\step; \ldots; \step] \cdot [\force] \cdot \Dirs'$: a
prefix without misspeculation, followed by a directive $\force$ that
initiates the misspeculation on a branch, and then by an arbitrary
suffix of directives. This neatly divides the proof into two: the
prefix phase is identical to the first case up to the point where
misspeculation begins, and the two ideal runs have reduced to the same
command.
\autoref{lem:faslh-ideal-noninterference} shows that the states
reached after the pivot directive $\force$ are equivalent.
%
%
At this point, we apply \autoref{lem:gilles-lemma} with the suffix
$\Dirs'$ to conclude the proof.
\end{proof}



Now we are ready to prove our main theorem.

\newtheorem*{thm:faslh-rs}{\autoref{thm:faslh-rs}}
\begin{thm:faslh-rs}[\FlexiSLH enforces relative security]
\setcounter{equation}{0}
\begin{align}
 &\cvar{b} \notin \UsedVars{\cexpr{c}} \wedge
  \lookup{\ScalarState_1}{\cvar{b}} = 0 \wedge
  \lookup{\ScalarState_2}{\cvar{b}} = 0
 &\Rightarrow
 \label{eq:faslh-general}
\\
 &(\forall \cvar{a}. | \cvar{a} |_{\ArrayState_1} > 0) \wedge
  (\forall \cvar{a}. | \cvar{a} |_{\ArrayState_2} > 0)
 &\Rightarrow
  \label{eq:faslh-nonempty}
\\
 &\colSel{
  \wtifc{\LabelPublic}{\cexpr{c}} \wedge
  \ScalarState_1 \sim_{\PubVars} \ScalarState_2 \wedge
  \ArrayState_1 \sim_{\PubArrs} \ArrayState_2
  }
 &\colSel{
  \Rightarrow
  }
  \label{eq:faslh-loweq}
\\
 &\colU{
 \seqstate{\cexpr{c}}{\ScalarState_1}{\ArrayState_1}
 \approx
 \seqstate{\cexpr{c}}{\ScalarState_2}{\ArrayState_2}
 }
 &\colU{\Rightarrow}
 \label{eq:faslh-seqeq}
\\
 &\specstate{\transls{\cexpr{c}}{\FlexiSLH}}{\ScalarState_1}{\ArrayState_1}{\BoolFalse}
  \approx_s
  \specstate{\transls{\cexpr{c}}{\FlexiSLH}}{\ScalarState_2}{\ArrayState_2}{\BoolFalse}
  \label{eq:faslh-goal}
\end{align}
\end{thm:faslh-rs}

\begin{proof}
The proof of relative security depends on some general assumptions
in \eqref{eq:faslh-general} about the original program not using the
reserved variable for the misspeculation flag, and the flag being
initialized to 0 in both memories.
Together with these, there is well-typedness and public-equivalence
in \eqref{eq:faslh-loweq} and the requirement on array sizes
in \eqref{eq:faslh-nonempty}.

The top-level proof itself is simple. \autoref{def-specobs} unfolds
in \eqref{eq:faslh-goal} to reveal a pair of speculative
executions. \autoref{lem:bcc} is used on each of these to yield a pair
of ideal executions. The result then follows
from \autoref{lem:faslh-ideal-rs}.
\end{proof}

The other two top-level security results for other variants of \SLH
are simple corollaries of this theorem.
Observe that the assumptions of \autoref{thm:faslh-rs} can be divided
into general-purpose assumptions (in black), assumptions that are
closely related to the security of \SeliSLH in \autoref{thm:saslh-ct}
(in \colSel{blue}), and those related to the security of \USLH
in \autoref{thm:uslh-rs} (in \colU{red}).



\newtheorem*{thm:saslh-ct}{\autoref{thm:saslh-ct}}
\begin{thm:saslh-ct}[\SeliSLH enforces SCT security]

\begin{align*}
  &\cvar{b} \notin \UsedVars{\cexpr{c}} \wedge
   \lookup{\ScalarState_1}{\cvar{b}} = 0 \wedge
   \lookup{\ScalarState_2}{\cvar{b}} = 0
  &\Rightarrow
 \\
  &(\forall \cvar{a}. | \cvar{a} |_{\ArrayState_1} > 0) \wedge
   (\forall \cvar{a}. | \cvar{a} |_{\ArrayState_2} > 0)
  &\Rightarrow
 \\
  &
   \wtct{\cexpr{c}} \wedge
   \ScalarState_1 \sim_{\PubVars} \ScalarState_2 \wedge
   \ArrayState_1 \sim_{\PubArrs} \ArrayState_2
  &
   \Rightarrow
 \\
  &\specstate{\transls{\cexpr{c}}{\SeliSLH}}{\ScalarState_1}{\ArrayState_1}{\BoolFalse}
   \approx_s
   \specstate{\transls{\cexpr{c}}{\SeliSLH}}{\ScalarState_2}{\ArrayState_2}{\BoolFalse}
 \end{align*}
\end{thm:saslh-ct}

\begin{proof}
From \autoref{thm:faslh-aslh}, \SeliSLH is identical to \FlexiSLH for
a fixed $\PubVars$.
Together with the use of the more restrictive \CCT type system instead
of the general IFC type system, this imposes that all observations are
based on public values, and we can deduce that all sequential
executions from public-equivalent states produce the same
observations. After deriving the missing assumption, we conclude by
applying \autoref{thm:faslh-rs}.
%
\end{proof}


\newtheorem*{thm:uslh-rs}{\autoref{thm:uslh-rs}}
\begin{thm:uslh-rs}[\USLH enforces relative security]
\begin{align*}
  &\cvar{b} \notin \UsedVars{\cexpr{c}} \wedge
   \lookup{\ScalarState_1}{\cvar{b}} = 0 \wedge
   \lookup{\ScalarState_2}{\cvar{b}} = 0
  &\Rightarrow
 \\
  &
   (\forall \cvar{a}. | \cvar{a} |_{\ArrayState_1} > 0) \wedge
   (\forall \cvar{a}. | \cvar{a} |_{\ArrayState_2} > 0)
  &
   \Rightarrow
 \\
  &\seqstate{\cexpr{c}}{\ScalarState_1}{\ArrayState_1}
   \approx
   \seqstate{\cexpr{c}}{\ScalarState_2}{\ArrayState_2}
  &\Rightarrow
 \\
  &\specstate{\translAux{\cexpr{c}}{\USLH}{}}{\ScalarState_1}{\ArrayState_1}{\BoolFalse}
   \approx_s
   \specstate{\translAux{\cexpr{c}}{\USLH}{}}{\ScalarState_2}{\ArrayState_2}{\BoolFalse}
 \end{align*}
\end{thm:uslh-rs}
\begin{proof}
From \autoref{thm:faslh-uslh}, \USLH is a special case of \FlexiSLH
when all variables are secret.
%
In this setting it is trivial to establish
$\wtifcAux{(\lambda \_. \LabelSecret)}{(\lambda \_. \LabelSecret)}{\LabelPublic}{\cexpr{c}}$,
$\ScalarState_1 \sim_{(\lambda \_. \LabelSecret)} \ScalarState_2$ and
$\ArrayState_1 \sim_{(\lambda \_. \LabelSecret)} \ArrayState_2$
for any choice of commands and states.
%
\autoref{thm:faslh-rs} then gives us relative security
for \USLH.
\end{proof}








\section{Flexible Value SLH}
\label{sec:FvSLH}

In the last section we formalized index \SLH, which protects
programs from speculative leakage by masking (some of) the indices
used in array accesses.
For load operations specifically, an alternative is to target the
``output'' of operations, preventing values read from memory
from leaking to the attacker. This has produced countermeasures like
Selective value \SLH \cite{ShivakumarBBCCGOSSY23}.

\begin{figure}
\centering
\begin{multline*}
\transl{\caread{\cvar{X}}{\cvar{a}}{\cexpr{i}}} \doteq
\\
\quad
  \left\{
  \begin{matrix*}[l]
  {\cseq{\caread{\cvar{X}}{\cvar{a}}{\cexpr{i}}}{\casgn{\cvar{X}}{\ccond{\cvar{b==1}}{\cvar{0}}{\cvar{X}}}}} &
  {\mathit{if} ~ \ValueCheck{\cvar{X}}{\cexpr{i}}}
  \\
  {\caread{\cvar{X}}{\cvar{a}}{\llbracket \cexpr{i} \rrbracket_{\mathit{rd}}}} &
  {\mathit{otherwise}}
  \end{matrix*}
  \right.
\end{multline*}
\vspace{-1em}
\caption{Master recipe for value \SLH}
\label{fig:vslh-template}
\end{figure}

\begin{figure}
\subcaptionbox{\SelectivevSLH \label{fig:svslh-instance}}[0.4\linewidth]{
\begin{align*}
\ValueCheck{\cvar{X}}{\cexpr{i}}
 &\doteq \colSel{
  \lookup{\PubVars}{\cvar{X}}
  }
\\
\llbracket \cexpr{i} \rrbracket_{\mathit{rd}}
 &\doteq \colSel{
  \phantom{\left\{ \right.}
  \begin{matrix*}[l]
  \phantom{\ccond{\cvar{b==1}}{\cvar{0}}{\cexpr{i}}}
  \\
  \cexpr{i}
  \end{matrix*}
  }
\\
\llbracket \cexpr{i} \rrbracket_{\mathit{wr}}
 &\doteq \colSel{
  \phantom{\left\{ \right.}
  \begin{matrix*}[l]
  \phantom{\ccond{\cvar{b==1}}{\cvar{0}}{\cexpr{i}}}
  \\
  \cexpr{i}
  \end{matrix*}
  }
\end{align*}
}
\subcaptionbox{Flexible \vSLH \label{fig:fvslh-instance}}[0.6\linewidth]{
\begin{align*}
 &\colSel{
  \lookup{\PubVars}{\cvar{X}}} \mathrel{\colFlexnew{\wedge}} \colFlexnew{\ArithLabel{\cexpr{i}}
  }
\\
 &\colFlexnew{
  \left\{
  \begin{matrix*}[l]
  {\colAll{\ccond{\cvar{b==1}}{\cvar{0}}{\cexpr{i}}}} &  {\mathit{if} ~ \colFlexnew{\lnot\labelarith{\PubVars}{\cexpr{i}}}}
  \\
  {\colSel{\cexpr{i}}} &
  {\mathit{otherwise}}
  \end{matrix*}
  \right.
  }
\\
 &\colFlexnew{
  \left\{
  \begin{matrix*}[l]
  {\colAll{\ccond{\cvar{b==1}}{\cvar{0}}{\cexpr{i}}}} &
  {\mathit{if} ~ \colFlexnew{\lnot\labelarith{\PubVars}{\cexpr{i}}}}
  \\
  {\colSel{\cexpr{i}}} &
  {\mathit{otherwise}}
  \end{matrix*}
  \right.
  }
\end{align*}
}
\caption{Overview of \vSLH variants}
\label{fig:vslh-variants}
\end{figure}




We can adapt \FlexibleSLH to use value \SLH countermeasures while
generalizing existing schemes and offering similar protections against
speculative attackers.
Our general template for \vSLH transformations is almost identical to
the \iSLH master recipe. Only the rule for array reads changes as
shown in \autoref{fig:vslh-template}.
In \iSLH, array reads were
parameterized by an index translation function
$\llbracket \cdot \rrbracket_{\mathit{rd}}^\cdot$ that was in charge of
masking the index as needed.  In \vSLH, the translation of reads is
split into two cases: one where the loaded value is immediately masked
using the misspeculation flag, and one where it is not; the rule is
parameterized by a value check function $\ValueCheck{\cdot}{\cdot}$
that determines which case is used.

The \vSLH template uses an index
translation function $\llbracket \cdot \rrbracket_{\mathit{rd}}$ in
the case where the read value is not masked, as well as the index
masking function $\llbracket \cdot \rrbracket_{\mathit{wr}}$ for array
writes, like \iSLH did.
%
Unlike \SelvSLH, which gets its security exclusively from masking
values loaded to public variables and never protects indices (see \autoref{fig:svslh-instance}),
\FlexvSLH (\autoref{fig:fvslh-instance})
needs to fall back on masking certain indices if it is to achieve
relative security.
Moreover, \USLH can also be derived from \FlexvSLH by simply ignoring
its value masking facilities and setting $\ValueCheck{\_}{\_}
= \BoolFalse$ to fall back on index masking for all
protections.
%




\begin{figure*}
  \centering
  \[
  \infer[Ideal\_Read]
  {\labelarith{\PubVars}{\cexpr{ie}} = \ell_{\cexpr{i}} \\
   i = \left\{
   {\begin{matrix*}[l]
    0 & \mathit{if} ~ \highlight{\lnot\ell_{\cexpr{i}} \mathrel{\wedge} b} \\
    \eval{\cexpr{ie}}{\rho} & \mathit{otherwise}
    \end{matrix*}
   }
   \right.
   \\\\
   v =
   \highlight{
   \left\{
   {\begin{matrix*}[l]
    0 & \mathit{if} ~ \lookup{\PubVars}{\cvar{X}} \mathrel{\wedge} \ell_{\cexpr{i}} \mathrel{\wedge} b \\
    \colAll{\eval{\cvar{a}[i]}{\mu}} & \mathit{otherwise}
    \end{matrix*}
   }
   \right.
   } \\
   i < | \cvar{a} |_{\mu}
  }
  {\IdealEval{\caread{\cvar{X}}{\cvar{a}}{\cexpr{ie}}}{\rho}{\mu}{b}{\cskip}{\subst{\cvar{X}}{v}{\rho}}{\mu}{b}{\OARead{\cvar{a}}{i}}{\step}}
  \quad
  \infer[Ideal\_Read\_Force]
  {\labelarith{\PubVars}{\cexpr{ie}} \\
   \highlight{\cancel{\lnot\lookup{\PubVars}{\cvar{X}}}} \\
   i = \eval{\cexpr{ie}}{\rho} \\\\
   v =
   \highlight{
   \left\{
   {\begin{matrix*}[l]
    0 & \mathit{if} ~ \lookup{\PubVars}{\cvar{X}} \\
    \colAll{\eval{\cvar{b}[j]}{\mu}} & \mathit{otherwise}
    \end{matrix*}
   }
   \right.
   } \\
   i \ge | \cvar{a} |_{\mu} \\
   j < | \cvar{b} |_{\mu}
  }
  {\IdealEval{\caread{\cvar{X}}{\cvar{a}}{\cexpr{ie}}}{\rho}{\mu}{\BoolTrue}{\cskip}{\subst{\cvar{X}}{v}{\rho}}{\mu}{\BoolTrue}{\OARead{\cvar{a}}{i}}{\DLoad{\cvar{b}}{j}}}
  \]
  \[
  \infer[Ideal\_Write]
  {
   i = \left\{
   {\begin{matrix*}[l]
    0 & \mathit{if} ~ \highlight{\lnot\ell_{\cexpr{i}} \mathrel{\wedge} b} \\
    \eval{\cexpr{ie}}{\rho} & \mathit{otherwise}
    \end{matrix*}
   }
   \right. \\\\
   \labelarith{\PubVars}{\cexpr{ie}} = \ell_{\cexpr{i}} \\
   \highlight{\cancel{\labelarith{\PubVars}{\cexpr{ae}} = \ell}} \\
   v = \eval{\cexpr{ae}}{\rho} \\ 
   i < | \cvar{a} |_{\mu}
  }
  {\IdealEval{\cawrite{\cvar{a}}{\cexpr{ie}}{\cexpr{ae}}}{\rho}{\mu}{b}{\cskip}{\rho}
  {\subst{\cvar{a}[i]}{v}{\mu}}
  {b}{\OAWrite{\cvar{a}}{i}}{\step}}
  \quad
  \infer[Ideal\_Write\_Force]
  {\labelarith{\PubVars}{\cexpr{ie}} \\
  \highlight{\cancel{\labelarith{\PubVars}{\cexpr{ae}}}} \\\\
   i = \eval{\cexpr{ie}}{\rho} \\
   v = \eval{\cexpr{ae}}{\rho} \\  
   i \ge | \cvar{a} |_{\mu} \\
   j < | \cvar{b} |_{\mu}
  }
  {\IdealEval{\cawrite{\cvar{a}}{\cexpr{ie}}{\cexpr{ae}}}{\rho}{\mu}{\BoolTrue}{\cskip}{\rho}
  {\subst{\cvar{b}[j]}{v}{\mu}}
  {\BoolTrue}{\OAWrite{\cvar{a}}{i}}{\DStore{\cvar{b}}{j}}}
  \]
  \vspace{-1em}
  \caption{Ideal semantics for \FlexvSLH (selected rules)}
  \label{fig:ideal-semantics-fvslh}
\end{figure*}

\begin{example}
\label{ex:fvslh-arrays-differ}
Consider again \autoref{ex:aslh-vs-vslh}.
Both \iSLH and \vSLH prevent the public variable \cvar{x} from containing the secret value of \cvar{key}, but their strategies differ.
In the \iSLH setting, it is the out-of-bounds store that is prevented. 
\vSLH, however, does not protect the write operation as long as the index $\cvar{i}$ is public. The attacker is indeed able to write the secret $\cvar{key}$ into the public array $\cvar{a}$,
which breaks the public-equivalence between arrays. In order to maintain security, the subsequently loaded value on $\cvar{x}$ is masked to 0, and \cvar{a[0]} containing a secret does not cause the program to leak information.
\end{example}


Using these new definitions we can state and prove that
\FlexvSLH enforces relative security and
\SelvSLH enforces SCT security.
The high-level structure of the security proofs closely mirrors the
development in \autoref{sec:formal-results}.
\FlexvSLH uses a new version of the ideal semantics that
reflects the changes in its behavior and supplies the proofs with
relevant information.
By a slight abuse of notation, we will write
$\PubVars \vdash \IdealEval{\cexpr{c}}{\rho}{\mu}{b}{\cexpr{c'}}{\rho'}{\mu'}{b'}{\Obss}{\Dirs}$
for the new semantics in this section and elide the public variables
and turnstile as those remain constant throughout.

The new ideal semantics is in large part identical to the
one for \FlexiSLH in \autoref{fig:ideal-semantics-faslh}; the
only changes are in the premises of the array read
and write rules, which are given in \autoref{fig:ideal-semantics-fvslh}.
Changes between the two versions are \highlight{highlighted};
premises that disappear in the new version are
also crossed out.
Both rules for array reads now apply masking to the value fetched from the array,
consistent with the
transformation. Additionally, both $\step$ rules on array reads and
writes have modified conditions for index masking, and both
misspeculating rules on reads and writes have slightly less restrictive premises.
%




Based on the new ideal semantics we can prove an identical set of key
technical lemmas, elided here\ifappendix~(see \autoref{sec:appendix-fvslh-theorems} for details)\fi.
The fact that arrays are no longer public-equivalent during
misspeculation, as \autoref{ex:fvslh-arrays-differ} shows, requires
us to strengthen the value-based counterpart
of \autoref{lem:gilles-lemma} by removing that assumption,
as well as an adapted version of \autoref{lem:faslh-ideal-noninterference}.
%
The statements of the main theorems correspond to their index \SLH counterparts.

\begin{theorem}
\label{thm:fvslh-rs}
\FlexvSLH enforces relative security for all IFC-well-typed programs and for all
public-equivalent initial states.
\end{theorem}



\begin{theorem}
\label{thm:svslh-sct}
\SelvSLH enforces SCT security for CCT programs.
%
\end{theorem}


\section{\FsFlexvSLH: Flexible vSLH for All Programs}
\label{sec:fsfvslh}

So far, we used the simple IFC type system of \autoref{sec:type-system}.
As outlined in \autoref{sec:fsfvslh-key-ideas},
we demonstrate that this is not a conceptual limitation by introducing \FsFlexvSLH,
which works on arbitrary source programs by using a sound but permissive static IFC analysis
similar to the algorithmic version of \citepos{HuntS06} flow-sensitive type system.
Unlike the IFC type system of \autoref{sec:type-system}, which rejects some programs
since it requires the labelings
$\PubVars$ and $\PubArrs$ to remain constant
, the flow-sensitive IFC analysis
accepts all programs by updating the labelings with each assignment.
%
This results in different labelings at different program points
that are used to compute IFC levels of individual program expressions.
Unlike \citet{HuntS06}, our analysis adds these levels as annotations to the commands,
as they are used by \FsFlexvSLH to apply protections.
We begin this section
by describing annotations and the IFC analysis,
before presenting the ideal semantics and well-labeledness predicate
that allow us to prove relative security for \FsFlexvSLH.

\subsection{Annotations and Flow-Sensitive IFC analysis}
\label{sec:acom-and-analysis}

\begin{figure}
  \centering
  \[
    \begin{aligned}
      \annotated{\cexpr{c}}\in\texttt{acom} ::=~ & \cskip ~|~ \casgn{\cvar{X}}{\cexpr{e}}
      ~|~\acseq{\annotated{\cexpr{c}}}{\annotated{\cexpr{c}}}{\PubVars}{\PubArrs}\\
      |&~\acifs{\cexpr{b}}{\annotated{\cexpr{c}}}{\annotated{\cexpr{c}}}{\ell}\\
      |&~\acwhiles{\cexpr{b}}{\annotated{\cexpr{c}}}{\ell}{\PubVars}{\PubArrs}\\
      |&~\acaread{\cvar{X}}{\cvar{a}}{\cexpr{e}}{\ell_\cvar{X}}{\ell_i}
      ~|~\acawrite{\cvar{a}}{\cexpr{e}}{\cexpr{e}}{\ell_i}\\
      |&~\acbranch{\ell}{\annotated{\cexpr{c}}}
    \end{aligned}
  \]
  \vspace{-0.5em}
  \caption{Annotated commands}
  \label{fig:syntax-acom}
\end{figure}

We extend the syntax of \SourceLang commands as shown in \autoref{fig:syntax-acom} with the following annotations:
\begin{itemize}[leftmargin=*,nosep]
  \item Branch conditions, array indices, and target variables of array reads are
    labeled with their statically determined security level.
    These annotations are directly used by \FsFlexvSLH.
  \item Sequence and loop commands are annotated with intermediate labelings; used for
defining well-labeledness (\autoref{sec:fs-well-labeled}).
  \item Finally, we introduce a new command $\acbranch{\ell}{\annotated{\cexpr{c}}}$, which wraps
$\annotated{\cexpr{c}}$ with an additional label $\ell$.
    This annotation is not produced by the IFC analysis, but by our ideal
    semantics of \autoref{sec:ideal-sem}, which also works with annotated commands,
    and which uses this annotation to better track implicit flows.
\end{itemize}

\begin{figure*}
  \begin{align*}
    \flowtrack{\cskip}{\PubVars}{\PubArrs}{\mathpc} \doteq& (\cskip, \PubVars, \PubArrs)\\
    \flowtrack{\casgn{\cvar{X}}{\cexpr{e}}}{\PubVars}{\PubArrs}{\mathpc} \doteq& (\casgn{\cvar{X}}{\cexpr{e}}, \subst{\cvar{X}}{\labelarith{\PubVars}{\cexpr{e}}}{\PubVars}, \PubArrs)\\
    \flowtrack{\cseq{\cexpr{c_1}}{\cexpr{c_2}}}{\PubVars}{\PubArrs}{\mathpc} \doteq& (\acseq{\annotated{\cexpr{c_1}}}{\annotated{\cexpr{c_2}}}{\PubVars_1}{\PubArrs_1}, \PubVars_2, \PubArrs_2) 
    \begin{array}[t]{rl}
      \text{where } & (\annotated{\cexpr{c_1}}, \PubVars_1, \PubArrs_1) = \flowtrack{\cexpr{c_1}}{\PubVars}{\PubArrs}{\mathpc} \\
      \text{and } & (\annotated{\cexpr{c_2}}, \PubVars_2, \PubArrs_2) = \flowtrack{\cexpr{c_2}}{\PubVars_1}{\PubArrs_1}{\mathpc}
    \end{array} \\
    \flowtrack{\cifs{\cexpr{be}}{\cexpr{c_1}}{\cexpr{c_2}}}{\PubVars}{\PubArrs}{\mathpc} \doteq& (\acifs{\cexpr{be}}{\annotated{\cexpr{c_1}}}{\annotated{\cexpr{c_2}}}{\labelarith{\PubVars}{\cexpr{be}}}, \PubVars_1 \sqcup \PubVars_2, \PubArrs_1 \sqcup \PubArrs_2)
    \begin{array}[t]{rl}
      \text{where } & (\annotated{\cexpr{c_1}}, \PubVars_1, \PubArrs_1) = \flowtrack{\cexpr{c_1}}{\PubVars}{\PubArrs}{\mathpc \sqcup \labelarith{\PubVars}{\cexpr{be}}}\\
      \text{and } & (\annotated{\cexpr{c_2}}, \PubVars_2, \PubArrs_2) = \flowtrack{\cexpr{c_2}}{\PubVars}{\PubArrs}{\mathpc \sqcup \labelarith{\PubVars}{\cexpr{be}}}
    \end{array}\\
    \flowtrack{\cwhiles{\cexpr{be}}{\cexpr{c}}}{\PubVars}{\PubArrs}{\mathpc} \doteq& (\acwhiles{\cexpr{be}}{\annotated{\cexpr{c}}}{\labelarith{\PubVars_{\mathit{fix}}}{\cexpr{be}}}{\PubVars_{\mathit{fix}}}{\PubArrs_{\mathit{fix}}}, \PubVars_{\mathit{fix}}, \PubArrs_{\mathit{fix}})
                                                                                 \\& \text{where } (\PubVars_{\mathit{fix}}, \PubArrs_{\mathit{fix}}) = \textbf{fix}~(\lambda (\PubVars', \PubArrs').\, \textbf{let}~(\annotated{\cexpr{c}}, \PubVars'', \PubArrs'') = \flowtrack{\cexpr{c}}{\PubVars'}{\PubArrs'}{\mathpc \sqcup \labelarith{\PubVars'}{\cexpr{be}}} ~ \textbf{in} ~ (\PubVars'', \PubArrs'') \sqcup (\PubVars, \PubArrs))\\
    \flowtrack{\caread{\cvar{X}}{\cvar{a}}{\cexpr{i}}}{\PubVars}{\PubArrs}{\mathpc} \doteq& (\acaread{\cvar{X}}{\cvar{a}}{\cexpr{i}}{\mathpc \sqcup \labelarith{\PubVars}{\cexpr{i}} \sqcup \PubArrs(\cvar{a})}{\labelarith{\PubVars}{\cexpr{i}}}, \subst{\cvar{X}}{\mathpc \sqcup \labelarith{\PubVars}{\cexpr{i}} \sqcup \PubArrs(\cvar{a})}{\PubVars}, \PubArrs)\\
    \flowtrack{\cawrite{\cvar{a}}{\cexpr{i}}{\cexpr{e}}}{\PubVars}{\PubArrs}{\mathpc} \doteq& (\acawrite{\cvar{a}}{\cexpr{i}}{\cexpr{e}}{\labelarith{\PubVars}{\cexpr{i}}}, \PubVars , \subst{\cvar{a}}{\PubArrs(\cvar{a}) \sqcup \mathpc \sqcup \labelarith{\PubVars}{\cexpr{i}} \sqcup \PubVars(\cexpr{e})}{\PubArrs})
  \end{align*}
  \vspace{-1.5em}
  \caption{Flow-sensitive IFC analysis generating annotated commands}
  \label{fig:flow-tracking}
\end{figure*}

The IFC analysis described by \autoref{fig:flow-tracking} is a straightforward extension
of the one of \citet{HuntS06} to \SourceLang and to also add annotations to the commands.
We write $(\PubVars_1, \PubArrs_1) \sqsubseteq (\PubVars_2, \PubArrs_2)$,
$\PubVars_1 \sqcup \PubVars_2$, and $\PubArrs_1 \sqcup \PubArrs_2$
for the pointwise liftings from labels to maps.
For assignments, loads, and stores, the analysis simply updates
the labelings and annotates expressions with the appropriate labels.
It also takes into account the program counter label $\mathpc$,
which as usual helps prevent implicit flows.
For conditionals, both branches are analyzed recursively,
raising the $\mathpc$ label by the label of the condition.
We then take the join of the two resulting labelings, which is a necessary
overapproximation, as we cannot determine statically which branch will be taken.

Similarly, for loops, we must account for arbitrarily many iterations.
This is achieved by a fixpoint labeling, \IE a labeling such that analysis
of the loop body with this labeling results in an equal or more precise final
labeling, which implies that the annotations produced by the analysis remain
correct for execution in this final labeling (where correctness of annotations
is formalized in \autoref{sec:fs-well-labeled}).
Further, this fixpoint labeling must also be less precise than the initial
labeling to ensure that the computed annotations are correct for the first iteration.
%
%
We compute the fixpoint labeling with a simple fixpoint iteration,
ensuring termination by additionally providing the number of variables and arrays assigned 
by $\cexpr{c}$ as an upper bound.
This is sound, as each iteration that does not yield a
fixpoint must change some array or variable from public to secret.
%
Thus, the static analysis is a total function that accepts all programs.

\subsection{Applying Protections}
The IFC analysis computes security labels for expressions and records them
as annotations in the commands.
The translation function $\translFsFlexvSLH{\annotated{\cexpr{c}}}$ works very
similarly to \FlexvSLH (cf. \autoref{fig:vslh-template},
\autoref{fig:fvslh-instance}) but takes annotated commands as input and applies
protections based on the labels provided as annotations, instead of the fixed
labelings for variables taken by \FlexvSLH.

\subsection{Ideal Semantics}
\label{sec:ideal-sem}

\begin{figure*}
  \centering
  \[
    \inferrule*[left=Ideal\_If]
    {\highlight{\cancel{\labelbool{\PubVars}{\cexpr{be}} = \ell}} \\
   b' = (\ell \vee \neg b) \mathrel{\wedge} \eval{\cexpr{be}}{\ScalarState}
  }
  {
    \idealstep{\fsidealstate{\highlight{\acifs{\cexpr{be}}{\annotated{\cexpr{c_\BoolTrue}}}{\annotated{\cexpr{c_\BoolFalse}}}{\ell}}}{\rho}{\mu}{b}{\highlight{\mathpc}}{\highlight\PubVars}{\highlight{\PubArrs}}}{\branch{b'}}{\step}{\fsidealstate{\highlight{\acbranch{\mathpc}{\annotated{\cexpr{c_{b'}}}}}}{\rho}{\mu}{b}{\highlight{\mathpc \sqcup \ell}}{\highlight{\PubVars}}{\highlight{\PubArrs}}}
}
  \]
  \[
  \inferrule*[left=Ideal\_If\_Force]
  {\highlight{\cancel{\BoolLabel{\cexpr{be}} = \ell}} \\
   b' = (\ell \vee \neg b) \mathrel{\wedge} \eval{\cexpr{be}}{\ScalarState}
  }
  {
    \idealstep{\fsidealstate{\highlight{\acifs{\cexpr{be}}{\annotated{\cexpr{c_\BoolTrue}}}{\annotated{\cexpr{c_\BoolFalse}}}{\ell}}}{\rho}{\mu}{b}{\highlight{\mathpc}}{\highlight{\PubVars}}{\highlight{\PubArrs}}}{\branch{b'}}{\force}{\fsidealstate{\highlight{\acbranch{\mathpc}{\annotated{\cexpr{c_{\neg b'}}}}}}{\rho}{\mu}{\BoolTrue}{\highlight{\mathpc \sqcup \ell}}{\highlight{\PubVars}}{\highlight{\PubArrs}}}
    }
  \]
  \[
  \inferrule*[left=Ideal\_Read]
  {\highlight{\cancel{\labelarith{\PubVars}{\cexpr{ie}} = \ell_{\cexpr{i}}} }\\
   i = \left\{
   {\begin{matrix*}[l]
    0 & \mathit{if} ~ \lnot\ell_{\cexpr{i}} \mathrel{\wedge} b \\
    \eval{\cexpr{ie}}{\rho} & \mathit{otherwise}
    \end{matrix*}
   }
   \right.
   \\
   v =
   \left\{
   {\begin{matrix*}[l]
       0 & \mathit{if} ~ \highlight{\ell_\cvar{X}} \mathrel{\wedge} \ell_{\cexpr{i}} \mathrel{\wedge} b \\
    \colAll{\eval{\cvar{a}[i]}{\mu}} & \mathit{otherwise}
    \end{matrix*}
   }
   \right.
    \\
   i < | \cvar{a} |_{\mu}
 }{
   \idealstep{\fsidealstate{\highlight{\acaread{\cvar{X}}{\cvar{a}}{\cexpr{ie}}{\ell_\cvar{X}}{\ell_\cexpr{i}}}}{\rho}{\mu}{b}{\highlight{\mathpc}}{\highlight{\PubVars}}{\highlight{\PubArrs}}}{\OARead{\cvar{a}}{i}}{\step}{\fsidealstate{\cskip}{\subst{\cvar{X}}{v}{\rho}}{\mu}{b}{\highlight{\mathpc}}{\highlight{\subst{\cvar{X}}{\ell_\cvar{X}}{\PubVars}}}{\highlight{\PubArrs}}}
 }
  \]
  \[
  \inferrule*[left=Ideal\_Read\_Force]
  {\highlight{\cancel{\labelarith{\PubVars}{\cexpr{ie}}}} \\
   i = \eval{\cexpr{ie}}{\rho} \\
   v =
   \left\{
   {\begin{matrix*}[l]
       0 & \mathit{if} ~ \highlight{\ell_\cvar{X}} \\
    \colAll{\eval{\cvar{b}[j]}{\mu}} & \mathit{otherwise}
    \end{matrix*}
   }
   \right.
    \\
   i \ge | \cvar{a} |_{\mu} \\
   j < | \cvar{b} |_{\mu} 
 }{
   \idealstep{\fsidealstate{\highlight{\acaread{\cvar{X}}{\cvar{a}}{\cexpr{ie}}{\ell_\cvar{X}}{\BoolTrue}}}{\rho}{\mu}{\BoolTrue}{\highlight{\mathpc}}{\highlight{\PubVars}}{\highlight{\PubArrs}}}{\OARead{\cvar{a}}{i}}{\step}{\fsidealstate{\cskip}{\subst{\cvar{X}}{v}{\rho}}{\mu}{\BoolTrue}{\highlight{\mathpc}}{\highlight{\subst{\cvar{X}}{\ell_\cvar{X}}{\PubVars}}}{\highlight{\PubArrs}}}
 }
  \]
  \[
  \inferrule*[lab=Ideal\_Seq\_Skip]
  {\highlight{\mathit{terminal} ~ \annotated{\cexpr{c_1}}}}
  {\idealstep{\fsidealstate{\highlight{\acseq{\annotated{\cexpr{c_1}}}{{\annotated{\cexpr{c_2}}}}{\PubVars'}{\PubArrs'}}}{\rho}{\mu}{b}{\highlight{\mathpc}}{\highlight{\PubVars}}{\highlight{\PubArrs}}}{\bullet}{\bullet}{\fsidealstate{\annotated{\cexpr{c_2}}}{\rho}{\mu}{b}{\highlight{\pcofacom{\annotated{\cexpr{c_1}}}{\mathpc}}}{\highlight{\PubVars}}{\highlight{\PubArrs}}}}
  \quad
    \highlight{
    \inferrule*[lab=Ideal\_Branch]
    {
      \idealstep{\fsidealstate{{\annotated{\cexpr{c}}}}{\rho}{\mu}{b}{\mathpc}{\PubVars}{\PubArrs}}{o}{d}{\fsidealstate{{\annotated{\cexpr{c'}}}}{\rho'}{\mu'}{b'}{\mathpc'}{\PubVars'}{\PubArrs'}}
    }
    {
      \idealstep{\langle \acbranch{\ell}{\annotated{\cexpr{c}}}, \dots \rangle}{o}{d}{\langle \acbranch{\ell}{\annotated{\cexpr{c'}}}, \dots \rangle}
    }
  }
  \]
  \vspace{-1.5em}
  \caption{Ideal semantics for \FsFlexvSLH (selected rules)}
  \label{fig:ideal-semantics-flowsensitivevslh}
\end{figure*}

As in \autoref{sec:ideal-aslh} and \autoref{sec:FvSLH},
we use an ideal semantics to prove the relative security of \FsFlexvSLH.
Compared to \autoref{fig:ideal-semantics-fvslh} from \autoref{sec:FvSLH},
the ideal semantics for \FsFlexvSLH 
also dynamically tracks the labels of variables when taking execution steps~\cite{SabelfeldR09},
as these will be necessary for the counterpart of \autoref{lem:gilles-lemma}.
Therefore, the states of our ideal semantics also include a $\mathpc$ label,
variable labeling $\PubVars$, and array labeling $\PubArrs$.
To exactly match the behavior of the statically applied protections,
the ideal semantics does not use these dynamic labelings in its premises though.
Instead, it uses the annotations provided by the static IFC analysis.
This distinction is important, as the dynamically computed labelings may be more precise,
depending on which branches are taken and the number of loop iterations.

To save and restore the $\mathpc$ label, the ideal semantics uses \texttt{branch} annotations.
When entering a branch $\annotated{\cexpr{c}}$ in rules \textsc{Ideal\_If} and
\textsc{Ideal\_If\_Force} in \autoref{fig:ideal-semantics-flowsensitivevslh},
$\annotated{\cexpr{c}}$ is wrapped in an annotated command
$\acbranch{\mathpc}{\annotated{\cexpr{c}}}$, where $\mathpc$ is the program
counter label before entering the branch.
The command $\annotated{\cexpr{c}}$ underneath this annotation
then takes steps as usual (\textsc{Ideal\_Branch}).
This leads to terminated commands ($\cskip$) potentially being wrapped in an arbitrary number of such branch annotations, so we introduce the predicate $\mathit{terminal}$ to describe such commands.
The rule \textsc{Ideal\_Seq\_Skip} is adjusted accordingly, allowing any terminal program on the left.
In this rule, the function $\pcofacom{\mathpc}{\annotated{\cexpr{c}}}$
determines the label of the outermost branch annotation and uses it to restore
the program counter label to the one before entering the branch.
If there is no branch annotation, $\pcofacom{\mathpc}{\annotated{\cexpr{c}}}$
defaults to the current program counter label $\mathpc$.
We choose to use branch annotations in this way as it allows us to keep a
one-to-one correspondence between steps of the sequential and ideal semantics.
%
For illustration, we also give the rules for loads in
\autoref{fig:ideal-semantics-flowsensitivevslh}, with differences
wrt. \autoref{fig:ideal-semantics-fvslh} \highlight{highlighted}.
\ifappendix The full set of rules can be found in \autoref{sec:appendix-fsflexvslh-ideal-semantics}.\fi

Since the ideal semantics matches the masking behavior of $\translFsFlexvSLH{\cdot}$,
we again establish a compiler correctness result:
\begin{lemma}[Backwards compiler correctness for \FsFlexvSLH]
  \label{lem:bcc-fs}
  \setcounter{equation}{0}
  \begin{align}
 &(\forall \cvar{a}. | \cvar{a} |_{\ArrayState} > 0) \wedge
 \cvar{b} \notin \UsedVars{\cexpr{c}} \wedge
 \lookup{\ScalarState}{\cvar{b}} = \natofbool{\SemSpecFlag}
 \\
 &\Rightarrow \specstate{\translFsFlexvSLH{\annotated{\cexpr{c}}}}{\ScalarState}{\ArrayState}{\SemSpecFlag}
 \specmulti{\Obss}{\Dirs}
 \specstate{\cexpr{c'}}{\ScalarState'}{\ArrayState'}{\SemSpecFlag'}
 \\
 &
\Rightarrow 
\setlength{\arraycolsep}{0pt}
\begin{array}[t]{rl}
\exists \annotated{\cexpr{c''}} \, \PubVars' \, \PubArrs' \, \mathpc'.&
  {\fsidealstate{\cexpr{c}}{\ScalarState}{\ArrayState}{\SemSpecFlag}{\mathpc}{\PubVars}{\PubArrs}}
  \xrightarrow[\Dirs]{\Obss}\hspace{-0.5em}{}_i^*\\
 & {\fsidealstate{\cexpr{c''}}{\subst{\cvar{b}}{\lookup{\ScalarState}{\cvar{b}}}{\ScalarState'}}{\ArrayState'}{\SemSpecFlag'}{\mathpc'}{\PubVars'}{\PubArrs'}}
\end{array}
\\
 &\wedge
 (
 \cexpr{c'} = \cskip \Rightarrow
 \mathit{terminal}~\annotated{\cexpr{c''}} \wedge
 \lookup{\ScalarState'}{\cvar{b}} = \natofbool{b'}
 )
  \end{align}
\end{lemma}

On the last line, we use a predicate
$\mathit{terminal} ~ \annotated{\cexpr{c''}}$ instead of requiring that
$\annotated{\cexpr{c''}} = \cskip$, because the $\cskip$ command may be wrapped
in multiple \texttt{branch} annotations.

\subsection{Well-labeledness}
\label{sec:fs-well-labeled}
\begin{figure}
  \centering
  \[
    \inferrule*[left=WL\_Skip]{ (\PubVars_1, \PubArrs_1) \sqsubseteq (\PubVars_2, \PubArrs_2)}{\wlacom{\cskip}{\PubVars_1}{\PubArrs_1}{\mathpc}{\PubVars_2}{\PubArrs_2}}
  \]
  \[
    \inferrule*[left=WL\_Asgn]{ (\subst{\cvar{X}}{\labelarith{\PubVars_1}{\cexpr{e}}}{\PubVars_1}, \PubArrs_1) \sqsubseteq{\PubVars_2, \PubArrs_2}}{ \wlacom{(\casgn{\cvar{X}}{\cexpr{e}})}{\PubVars_1}{\PubArrs_1}{\mathpc}{\PubVars_2}{\PubArrs_2}}
  \]
  \[
    \inferrule*[left=WL\_Seq]{
      \branchfree{\annotated{\cexpr{c_2}}} \\
      \wlacom{\annotated{\cexpr{c_1}}}{\PubVars_1}{\PubArrs_1}{\mathpc}{\PubVars'}{\PubArrs'} \\
      \wlacom{\annotated{\cexpr{c_2}}}{\PubVars'}{\PubArrs'}{(\pcofacom{\annotated{\cexpr{c_1}}}{\mathpc})}{\PubVars_2}{\PubArrs_2}
    }{\wlacom{(\acseq{\annotated{\cexpr{c_1}}}{\annotated{\cexpr{c_2}}}{\PubVars'}{\PubArrs'})}{\PubVars_1}{\PubArrs_1}{\mathpc}{\PubVars_2}{\PubArrs_2}}
  \]
  \[
    \infer[WL\_If] {
      \labelbool{\PubVars_1}{\cexpr{be}} \sqsubseteq \ell_\cexpr{be}\and
      \branchfree{\annotated{\cexpr{c_1}}}\and
      \branchfree{\annotated{\cexpr{c_2}}}\\
      \wlacom{\annotated{\cexpr{c_1}}}{\PubVars_1}{\PubArrs_1}{\mathpc \sqcup \ell_\cexpr{be}}{\PubVars_2}{\PubArrs_2}\\
      \wlacom{\annotated{\cexpr{c_2}}}{\PubVars_1}{\PubArrs_1}{\mathpc \sqcup \ell_\cexpr{be}}{\PubVars_2}{\PubArrs_2}
    }{
      \wlacom{(\acifs{\cexpr{be}}{\annotated{\cexpr{c_1}}}{\annotated{\cexpr{c_2}}}{\ell_\cexpr{be}})}{\PubVars_1}{\PubArrs_1}{\mathpc}{\PubVars_2}{\PubArrs_2}
    }
  \]
  \[
    \infer[WL\_While]{
      \labelbool{\PubVars_1}{\cexpr{be}} \sqsubseteq \ell_\cexpr{be}\and
      \branchfree{\annotated{\cexpr{c}}}\\
      (\PubVars_1, \PubArrs_1) \sqsubseteq (\PubVars', \PubArrs') \and
      (\PubVars', \PubArrs') \sqsubseteq (\PubVars_2, \PubArrs_2) \\
      \wlacom{\annotated{\cexpr{c_1}}}{\PubVars'}{\PubArrs'}{\mathpc \sqcup \ell_\cexpr{be}}{\PubVars'}{\PubArrs'}
    }{
      \wlacom{(\acwhiles{\cexpr{be}}{\annotated{\cexpr{c}}}{\ell_\cexpr{be}}{\PubVars'}{\PubArrs'})}{\PubVars_1}{\PubArrs_1}{\mathpc}{\PubVars_2}{\PubArrs_2}
    }
  \]
  \[
    \infer[WL\_ARead]{
      \labelarith{\PubVars_1}{\cexpr{e}} \sqsubseteq \ell_i \\
      \mathpc \sqsubseteq \ell_\cvar{X} \\
      \ell_i \sqsubseteq \ell_\cvar{X} \\
      \PubArrs_1(\cvar{a}) \sqsubseteq \ell_\cvar{X} \\
      (\subst{\cvar{X}}{\ell_\cvar{X}}{\PubVars_1}, \PubArrs_1) \sqsubseteq (\PubVars_2, \PubArrs_2)
    }{
      \wlacom{(\acaread{\cvar{X}}{\cvar{a}}{\cvar{e}}{\ell_\cvar{X}}{\ell_i})}{\PubVars_1}{\PubArrs_1}{\mathpc}{\PubVars_2}{\PubArrs_2}
    }
  \]
  \[
    \infer[WL\_AWrite]{
      \labelarith{\PubVars_1}{\cexpr{i}} \sqsubseteq \ell_\cexpr{i} \\
      (\PubVars_1, \subst{\cvar{a}}{\PubArrs_1(\cvar{a}) \sqcup \mathpc \sqcup \ell_\cexpr{i} \sqcup \labelarith{\PubVars_1}{\cexpr{e}}}{\PubArrs_1}) \sqsubseteq (\PubVars_2, \PubArrs_2)
    }{
      \wlacom{(\acawrite{\cvar{a}}{\cexpr{i}}{\cexpr{e}}{\ell_\cexpr{i}})}{\PubVars_1}{\PubArrs_1}{\mathpc}{\PubVars_2}{\PubArrs_2}
    }
  \]
  \[
    \inferrule*[left=WL\_Branch]{
      \wlacom{\annotated{\cexpr{c}}}{\PubVars_1}{\PubArrs_1}{\mathpc}{\PubVars_2}{\PubArrs_2}
    }{
      \wlacom{(\acbranch{\ell}{\annotated{\cexpr{c}}})}{\PubVars_1}{\PubArrs_1}{\mathpc}{\PubVars_2}{\PubArrs_2}
    }
  \]
  \vspace{-1.5em}
  \caption{Well-labeledness of annotated commands}
  \label{fig:well-labeled}
\end{figure}

Since the ideal semantics relies on the labels included in our annotated commands, we can only show relative security (the counterpart of \autoref{lem:faslh-ideal-rs}) if the annotations have been computed correctly, \IE a secret value can never be labeled as public.
While this is the case for the initial annotated program obtained by our IFC analysis, we need a property that is preserved during execution, and thus cannot reference the analysis directly.
We therefore define a well-labeledness judgment
$\wlacom{\annotated{\cexpr{c}}}{\PubVars_1}{\PubArrs_1}{\mathpc}{\PubVars_2}{\PubArrs_2}$
of an annotated command $\annotated{\cexpr{c}}$ with respect to an initial
labeling $(\PubVars_1, \PubArrs_1)$ and $\mathpc$ label
and a final labeling $(\PubVars_2, \PubArrs_2)$ with the derivation rules given
in \autoref{fig:well-labeled}.
The rules largely match the behavior of
$\flowtrack{\cdot}{\PubVars}{\PubArrs}{\mathpc}$, except instead of saving
computed labels, we only make sure that they don't contradict the already annotated ones.

The most interesting case is that of loops (\textsc{WL\_While}).
Instead of doing another fixpoint iteration,
we confirm that the annotated labeling is indeed a fixpoint by using it for both
the initial and final labeling for the loop body.

Like in the ideal semantics, in \textsc{WL\_Seq}
we use $\pcofacom{\ell}{\annotated{\cexpr{c_1}}}$
to determine the correct program counter label after executing the first part of a sequence.
However, this would create problems for the preservation result below
if we had branch annotations at arbitrary locations, so we prohibit branch annotations in nonsensical locations 
using the predicate $\branchfree{\annotated{\cexpr{c}}}$.

Note that well-labeledness does not require labels to match \emph{precisely}, instead,
they may overapproximate.
Specifically, the initial labeling can be made more precise and the final labeling less precise while preserving well-labeledness.

\begin{lemma}[IFC analysis produces well-labeled programs]
  \label{lem:static-tracking-well-labeled}
  \[ \flowtrack{{\cexpr{c}}}{\PubVars}{\PubArrs}{\mathpc} = (\annotated{\cexpr{c}}, \PubVars', \PubArrs') \implies \wlacom{\annotated{\cexpr{c}}}{\PubVars}{\PubArrs}{\mathpc}{\PubVars'}{\PubArrs'}\]
\end{lemma}
\begin{proof}[Proof sketch]
  By induction on $\cexpr{c}$; most cases follow easily from the definition of well-labeledness.
  The loop case follows by a nested induction on the number of public variables and arrays, with two base cases: If there are no public variables or arrays left or if the set of public variables and arrays does not shrink in one iteration, then a fixpoint has been reached; in both cases, well-labeledness of the loop body in the fixpoint labeling follows by the outer induction hypothesis.
\end{proof}

\begin{lemma}[$\idealstep{}{}{}{}$ preserves well-labeledness]
  \label{lem:fs-ideal-preserves-wl}
  \begin{align*}
    & \wlacom{\annotated{\cexpr{c}}}{\PubVars}{\PubArrs}{\mathpc}{\PubVars''}{\PubArrs''} & \Rightarrow \\
    & \idealstep{{\fsidealstate{\annotated{\cexpr{c}}}{\rho}{\mu}{b}{\mathpc}{\PubVars}{\PubArrs}}}{\Obss}{\Dirs}{\fsidealstate{\annotated{\cexpr{c'}}}{\rho}{\mu}{b}{\mathpc'}{\PubVars'}{\PubArrs'}} & \Rightarrow \\
    & \qquad \wlacom{\annotated{\cexpr{c'}}}{\PubVars'}{\PubArrs'}{\mathpc'}{\PubVars''}{\PubArrs''} 
  \end{align*}
\end{lemma}
\begin{proof}[Proof sketch]
  By induction on the annotated command.
  The loop case relies on the fact that well-labeledness is preserved when making the initial labeling more precise, as the labeling when reaching the loop is generally more precise than the loop's fixed-point labeling.
The sequence case is similar.
\end{proof}

\subsection{Key Theorems}

We have already shown that \FsFlexvSLH is correct w.r.t. the ideal semantics (\autoref{lem:bcc-fs}), so it remains to show the relative security of the ideal semantics.
One difference compared to \autoref{lem:faslh-ideal-rs} is that we compare the
ideal semantics on annotated commands to the nonspeculative semantics of
unannotated commands, so in \eqref{eq:fsfvslh-flowtrack-eq} we assume that the
annotated command was generated using the IFC analysis on the source program
using the same initial labeling used for the ideal semantics.

\begin{lemma}[$\idealmultiarrow{}{}$ ensures relative security]
  \label{lem:fs-fsvslh-ideal-rs}
\setcounter{equation}{0}
\begin{align}
 &  \flowtrack{\cexpr{c}}{\PubVars}{\PubArrs}{\BoolTrue} = (\annotated{\cexpr{c}}, \PubVars', \PubArrs') \wedge
  \ScalarState_1 \sim_{\PubVars} \ScalarState_2 \wedge
  \ArrayState_1 \sim_{\PubArrs} \ArrayState_2
 &\Rightarrow
 \label{eq:fsfvslh-flowtrack-eq}
\\
 &\seqstate{\cexpr{c}}{\ScalarState_1}{\ArrayState_1}
  \approx
  \seqstate{\cexpr{c}}{\ScalarState_2}{\ArrayState_2}
 &\Rightarrow
\\
 &\fsidealstate{\annotated{\cexpr{c}}}{\ScalarState_1}{\ArrayState_1}{\BoolFalse}{\BoolTrue}{\PubVars}{\PubArrs}
 \approx_i
 \fsidealstate{\annotated{\cexpr{c}}}{\ScalarState_2}{\ArrayState_2}{\BoolFalse}{\BoolTrue}{\PubVars}{\PubArrs}
\end{align}
\end{lemma}
\begin{proof}[Proof sketch]
  We apply \autoref{lem:static-tracking-well-labeled} to obtain well-labeledness of $\annotated{\cexpr{c}}$.
  The rest of the proof is analogous to \autoref{lem:faslh-ideal-rs}, since we
  can prove counterparts to \autoref{lem:faslh-ideal-noninterference} and
  \autoref{lem:gilles-lemma} using well-labeledness instead of typing, and since
  well-labeledness is preserved by the ideal semantics
  (\autoref{lem:fs-ideal-preserves-wl}).
\end{proof}

We compose Lemmas \ref{lem:bcc-fs} and \ref{lem:fs-fsvslh-ideal-rs} to prove the relative security of \FsFlexvSLH.
While \autoref{lem:fs-fsvslh-ideal-rs} requires the annotated command to be the result of the IFC analysis, this analysis always succeeds, so we obtain relative security \emph{for all programs}.

\newtheorem*{thm:fsfvslh-rs}{\autoref{thm:fvslh-rs}}
\begin{thm:fsfvslh-rs}[Relative security of \FsFlexvSLH for all programs]
\begin{align*}
 &\cvar{b} \notin \UsedVars{\cexpr{c}} \wedge
  \lookup{\ScalarState_1}{\cvar{b}} = 0 \wedge
  \lookup{\ScalarState_2}{\cvar{b}} = 0
 &\Rightarrow
\\
 &(\forall \cvar{a}. | \cvar{a} |_{\ArrayState_1} > 0) \wedge
  (\forall \cvar{a}. | \cvar{a} |_{\ArrayState_2} > 0)
 &\Rightarrow
\\
 &
  \ScalarState_1 \sim_{\PubVars} \ScalarState_2 \wedge
  \ArrayState_1 \sim_{\PubArrs} \ArrayState_2
 &
  \Rightarrow
\\
 &
 \seqstate{\cexpr{c}}{\ScalarState_1}{\ArrayState_1}
 \approx
 \seqstate{\cexpr{c}}{\ScalarState_2}{\ArrayState_2}
 &\Rightarrow
\\
 &\flowtrack{\cexpr{c}}{\PubVars}{\PubArrs}{\BoolTrue} = (\annotated{\cexpr{c}}, \PubArrs', \PubVars') &\Rightarrow \\
 &\specstate{\translFsFlexvSLH{\annotated{\cexpr{c}}}}{\ScalarState_1}{\ArrayState_1}{\BoolFalse}
  \approx_s
  \specstate{\translFsFlexvSLH{\annotated{\cexpr{c}}}}{\ScalarState_2}{\ArrayState_2}{\BoolFalse}
\end{align*}

\end{thm:fsfvslh-rs}

\section{Related Work}
\label{sec:related-work}

SLH was originally proposed and implemented in
LLVM~\cite{Carruth18} as a defense against Spectre v1 attacks.
\citet{PatrignaniG21} noticed that SLH is not strong enough to achieve one of
their relative security notions and proposed Strong SLH, which uses the
misspeculation flag to also mask all branch conditions and all addresses of
memory accesses.
\citet{ZhangBCSY23} later noticed that the inputs of all variable-time
instructions also have to be masked for security and implemented everything as
the Ultimate SLH program transformation in the x86 backend of LLVM.
They experimentally evaluated Ultimate SLH on the SPEC2017 benchmarks and
reported overheads of around 150\%, which is large, but still smaller than
adding fences.

Taking inspiration in prior static analysis work~\cite{CheangRSS19,
  GuarnieriKMRS20}, both \citet{PatrignaniG21} and \citet{ZhangBCSY23} propose
to use relative security to assess the security of their transformations, by
requiring that the transformed program does not leak speculatively more than
{\em the transformed program itself} leaks sequentially.
%
While this is a sensible way to define relative security in the absence of a
program transformation~\cite{CheangRSS19, GuarnieriKMRS20, DongolGPW24,
CauligiDMBS22},
%
when applied only to the transformed program, this deems secure a transformation
that introduces sequential leaks, instead of removing speculative ones.
Instead, we use a security definition that is more suitable for transformations
and compilers, which requires that any transformed program running with
speculation does not leak more than what the {\em source} program leaks sequentially.
Finally, our definition only looks at 4 executions (as opposed to 8
executions~\cite{PatrignaniG21}) and our proofs are fully mechanized in Rocq.
%

In a separate line of research, \SelvSLH was proposed as an efficient way to
protect cryptographic code against Spectre v1 in the Jasmin
language~\cite{AlmeidaBBBGLOPS17}.
\citet{ShivakumarBBCCGOSSY23} proved in a simplified setting that an automatic
\SelvSLH transformation (discussed in \autoref{sec:FvSLH}) achieves speculative
CT security.
Yet for the best efficiency, Jasmin programmers may have to reorganize their code
and manually insert a minimal number SLH protections, with a static analysis
just checking that the resulting code is secure~\cite{ShivakumarBGLOPST23}.
In more recent work, \citet{OlmosBBGL25} showed how speculative CT security can
be preserved by compilation in a simplified model of the Jasmin compiler, which
they formalized in Rocq.
These preservation proofs seem easily extensible to the stronger relative
security definition we use in this paper, yet providing security guarantees to
non-CCT code does not seem a goal for the Jasmin language, which is specifically
targeted at cryptographic implementations.

The simple and abstract speculative semantics we use in \autoref{sec:defs} is
taken from the paper of \citet{ShivakumarBBCCGOSSY23}, who credit prior work by
\citet{CauligiDGTSRB20} and \citet{BartheCGKLOPRS21} for the idea of a
``forwards'' semantics that takes attacker directions and does no rollbacks.
This style of speculative semantics seems more suitable for higher-level
languages, for which the concept of a ``speculation window'' that triggers
rollbacks would require exposing too many low-level details about the
compilation and target architecture, if at all possible to accommodate.
\citet{BartheCGKLOPRS21} also prove that in their setting, any speculative CT
attack against a semantics with rollbacks also exist in their ``forwards''
semantics with directions.

\section{Conclusion and Future Work}
\label{sec:conclusion}\label{sec:future-work}

In this paper we introduced FSLH, a flexible variant of SLH that brings together
the benefits of both Selective SLH and Ultimate SLH.
We provide formal proofs in the Rocq prover that FSLH satisfies a strong
relative security property, ensuring that the hardened program can only
leak speculatively as much as the original source program leaks sequentially.


On the practical side we would like to implement FSLH in LLVM and experimentally
evaluate the reduction in overhead with respect to Ultimate SLH.
This raises significant engineering challenges though, since it requires adding
a flow-sensitive IFC analysis to LLVM, keeping track of which program inputs are
secret and which ones not throughout the whole compiler chain, and
making sure that the defenses we add are not removed by subsequent compiler passes.
Finally, while Rocq proofs for the whole of LLVM would be way beyond what's
reasonable in terms of effort, we believe that property-based testing of
relative security could be a pragmatic compromise for validating end-to-end
security in practice.

Another interesting direction would be to test the existing implementations of
SLH in LLVM, which are much more complex than the simple iSLH and vSLH variants
of this paper, and for which security has not been investigated formally.

\ifanon\else
{\small
\paragraph{Acknowledgments}
We are grateful to Sven Argo, Santiago Arranz Olmos, Gilles Barthe, Lucie Lahaye, and
Zhiyuan Zhang for the insightful discussions.
We also thank the CSF'25 reviewers for their helpful feedback.
This work was in part supported by the DFG
as part of the Excellence Strategy of the German Federal and State Governments
-- EXC 2092 CASA -- 390781972.
\ch{TODO: add more acknowledgements here, if needed}
}
\fi


\ifappendix
\appendices

\section{Sequential Semantics for \SourceLang}
\label{sec:appendix-sequential-semantics}

 \autoref{fig:seq-semantics} shows the standard sequential semantics
for \SourceLang as presented
in \autoref{sec:defs-syntax-seq-semantics}. This figure can be derived
from \autoref{fig:spec-semantics} by erasing the highlighted
components.

\begin{figure}[!h]
  \centering \[
\inferrule*[left=Seq\_Asgn] {v = \eval{\cexpr{ae}}{\rho}}
  {\SeqEval{\casgn{\cvar{X}}{\cexpr{ae}}}{\cexpr{\rho}}{\cexpr{\mu}}{\cskip}
  {\subst{\cvar{X}}{v}{\rho}}
  {\cexpr{\mu}}{\bullet}} \]\[
\inferrule*[left=Seq\_Seq\_Step]
  {\SeqEval{\cexpr{c_1}}{\rho}{\mu}{\cexpr{c_1'}}{\rho'}{\cexpr{\mu'}}{o}}
  {\SeqEval{\cseq{\cexpr{c_1}}{\cexpr{c_2}}}{\rho}{\mu}{\cseq{\cexpr{c_1'}}{\cexpr{c_2}}}{\rho'}{\mu'}{o}}\]\[
\inferrule*[left=Seq\_Seq\_Skip]
  {\quad}
  {\SeqEval{\cseq{\cskip}{\cexpr{c}}}{\rho}{\mu}{\cexpr{c}}{\rho}{\mu}{\bullet}} \] \[ \infer[Seq\_If]
  {b' = \eval{\cexpr{be}}{\cexpr{\rho}}}
  {\SeqEval{\cifs{\cexpr{be}}{\cexpr{c_\BoolTrue}}{\cexpr{c_\BoolFalse}}}{\cexpr{\rho}}{\cexpr{\mu}}{\cexpr{c_{b'}}}{\cexpr{\rho}}{\mu}{\branch{b'}}}
  \]
  \[
  \infer[Seq\_While]
  {\cexpr{c_{while}} = \cwhiles{\cexpr{be}}{\cexpr{c}}}
  {\SeqEval{\cexpr{c_{while}}}{\rho}{\mu}{\cifs{\cexpr{be}}{\cseq{\cexpr{c}}{\cexpr{c_{while}}}}{\cskip}}{\rho}{\mu}{\bullet}}
  \]
  \[
  \inferrule*[left=Seq\_Read]
  {i = \eval{\cexpr{ie}}{\rho} \\
   v = \eval{\cvar{a}[i]}{\mu} \\
   i < | \cvar{a} |_{\mu}}
  {\SeqEval{\caread{\cvar{X}}{\cvar{a}}{\cexpr{ie}}}{\rho}{\mu}{\cskip}
   {\subst{\cvar{X}}{v}{\rho}}
   {\mu}{\OARead{\cvar{a}}{i}}}
  \]\[
  \inferrule*[left=Seq\_Write]
  {i = \eval{\cexpr{ie}}{\rho} \\
   v = \eval{\cexpr{ae}}{\rho} \\
   i < | \cvar{a} |_{\mu}}
  {\SeqEval{\cawrite{\cvar{a}}{\cexpr{ie}}{\cexpr{ae}}}{\rho}{\mu}{\cskip}{\rho}
   {\subst{\cvar{a}[i]}{v}{\mu}}
   {\OAWrite{\cvar{a}}{i}}}
  \]
  \caption{Sequential semantics}
  \label{fig:seq-semantics}
\end{figure}

\section{Full Theorem Statements for \FlexvSLH}
\label{sec:appendix-fvslh-theorems}

As we point out in \autoref{sec:FvSLH}, the statements and high-level
proofs for \FlexvSLH are very similar to their corresponding
address-based variants. This starts with the relation
between \FlexibleSLH and the other protection schemes.

\begin{theorem}[Connection between \FlexvSLH and \SelvSLH]
\label{thm:fvslh-svslh}
Given a set of public variables $\PubVars$,
for any \CCT program $c$
\[
\transls{\cexpr{c}}{\FlexvSLH} = \transls{\cexpr{c}}{\SelvSLH}
\]
\end{theorem}

\begin{proof}
Immediate from \autoref{fig:svslh-instance}
and \autoref{fig:fvslh-instance}.
\end{proof}

\begin{theorem}[Connection between \FlexvSLH and \USLH]
\label{thm:fvslh-uslh}
If all variables of a program (both scalars and arrays) are considered
secret, \IE $(\lambda \_. \LabelSecret)$, then for any program $c$
\[
\translAux{\cexpr{c}}{\FlexvSLH}{(\lambda \_. \LabelSecret)}
= \translAux{\cexpr{c}}{\USLH}{}
\]
\end{theorem}

\begin{proof}
Immediate from \autoref{fig:fvslh-instance}
and \autoref{fig:uslh-instance}, noting that \FlexvSLH is a strict
generalization of \FlexiSLH, to which it reduces when $\ValueCheck{\_}{\_}
= \BoolFalse$.
\end{proof}

The assumption of public-equivalence of array states that holds for \iSLH
no longer applies.
This change requires an adapted version
of \autoref{lem:faslh-ideal-noninterference}:

\begin{lemma}[Noninterference of $\idealsteparrow{}{}$ with equal observations]
  \label{lem:fvslh-ideal-noninterference}
  \begin{align*}
    &\wtifc{\mathpc}{\cexpr{c}}\wedge \ScalarState_1 \sim_{\PubVars} \ScalarState_2
    \wedge
    (b = \BoolFalse \Rightarrow \ArrayState_1 \sim_{\PubArrs} \ArrayState_2) &\Rightarrow\\
    &\IdealEval{\cexpr{c}}{\ScalarState_1}{\ArrayState_1}{b}{\cexpr{c_1}}{\ScalarState_1'}{\ArrayState_1'}{b_1}{o}{d}&\Rightarrow\\
    &\IdealEval{\cexpr{c}}{\ScalarState_2}{\ArrayState_2}{b}{\cexpr{c_2}}{\ScalarState_2'}{\ArrayState_2'}{b_2}{o}{d}&\Rightarrow\\
    &\cexpr{c_1} = \cexpr{c_2} \wedge b_1 = b_2 \wedge \ScalarState_1' \sim_{\PubVars} \ScalarState_2' \wedge (b_1 = \BoolFalse \Rightarrow \ArrayState_1' \sim_{\PubArrs} \ArrayState_2') &
  \end{align*}
\end{lemma}
\begin{proof}[Proof sketch]
  By induction on the first evaluation judgment and inversion on the second one.
\end{proof}

This lemma is used to prove the adaptation
of \autoref{lem:gilles-lemma} to the \vSLH setting, whose statement
also removes the assumption on the public-equivalence of arrays. The same
lemma is used to prove the main theorem:

\newtheorem*{thm:fvslh-rs}{\autoref{thm:fvslh-rs}}
\begin{thm:fvslh-rs}[\FlexvSLH enforces relative security]
\begin{align*}
 &\cvar{b} \notin \UsedVars{\cexpr{c}} \wedge
  \lookup{\ScalarState_1}{\cvar{b}} = 0 \wedge
  \lookup{\ScalarState_2}{\cvar{b}} = 0
 &\Rightarrow
\\
 &(\forall \cvar{a}. | \cvar{a} |_{\ArrayState_1} > 0) \wedge
  (\forall \cvar{a}. | \cvar{a} |_{\ArrayState_2} > 0)
 &\Rightarrow
\\
 &\colSel{
  \wtifc{\LabelPublic}{\cexpr{c}} \wedge
  \ScalarState_1 \sim_{\PubVars} \ScalarState_2 \wedge
  \ArrayState_1 \sim_{\PubArrs} \ArrayState_2
  }
 &\colSel{
  \Rightarrow
  }
\\
 &\colU{
 \seqstate{\cexpr{c}}{\ScalarState_1}{\ArrayState_1}
 \approx
 \seqstate{\cexpr{c}}{\ScalarState_2}{\ArrayState_2}
 }
 &\colU{\Rightarrow}
\\
 &\specstate{\transls{\cexpr{c}}{\FlexvSLH}}{\ScalarState_1}{\ArrayState_1}{\BoolFalse}
  \approx_s
  \specstate{\transls{\cexpr{c}}{\FlexvSLH}}{\ScalarState_2}{\ArrayState_2}{\BoolFalse}
\end{align*}
\end{thm:fvslh-rs}

\begin{proof}
Same structure as the proof of \autoref{thm:faslh-rs}.
\end{proof}

\newtheorem*{thm:svslh-sct}{\autoref{thm:svslh-sct}}
\begin{thm:svslh-sct}[\SelvSLH enforces SCT security \cite{ShivakumarBBCCGOSSY23}]
\begin{align*}
  &\cvar{b} \notin \UsedVars{\cexpr{c}} \wedge
   \lookup{\ScalarState_1}{\cvar{b}} = 0 \wedge
   \lookup{\ScalarState_2}{\cvar{b}} = 0
  &\Rightarrow
 \\
  &(\forall \cvar{a}. | \cvar{a} |_{\ArrayState_1} > 0) \wedge
   (\forall \cvar{a}. | \cvar{a} |_{\ArrayState_2} > 0)
  &\Rightarrow
 \\
  &
   \wtct{\cexpr{c}} \wedge
   \ScalarState_1 \sim_{\PubVars} \ScalarState_2 \wedge
   \ArrayState_1 \sim_{\PubArrs} \ArrayState_2
  &
   \Rightarrow
 \\
  &\specstate{\transls{\cexpr{c}}{\SelvSLH}}{\ScalarState_1}{\ArrayState_1}{\BoolFalse}
   \approx_s
   \specstate{\transls{\cexpr{c}}{\SelvSLH}}{\ScalarState_2}{\ArrayState_2}{\BoolFalse}
 \end{align*}
\end{thm:svslh-sct}

\begin{proof}
Simple corollary of \autoref{thm:fvslh-svslh}
and \autoref{thm:fvslh-rs}.
\end{proof}

\setcounter{theorem}{0}
\renewcommand{\thetheorem}{\thesection.\arabic{theorem}}

Because \USLH is a special case of \FlexvSLH, we can also
use \autoref{thm:fvslh-uslh} and \autoref{thm:fvslh-rs} to obtain an
alternative proof of \autoref{thm:uslh-rs}.

\section{Full definitions for \FsFlexvSLH}
\label{sec:appendix-fsflexvslh-full}

\subsection{Translation Function}
\label{sec:appendix-fsflexvslh-transl}

Figure \autoref{fig:trans-fsflexvslh} shows the full \FsFlexvSLH translation function.
\begin{figure*}
  \begin{minipage}{0.7\textwidth}
  \begin{align*}
    \translFsFlexvSLH{\cskip} &\doteq \cskip \\
    \translFsFlexvSLH{\casgn{\cvar{x}}{\cvar{ae}}} &\doteq \casgn{\cvar{x}}{\cvar{ae}}\\
    \translFsFlexvSLH{\acseq{\annotated{\cexpr{c_1}}}{\annotated{\cexpr{c_2}}}{\PubVars}{\PubArrs}} &\doteq \cseq{\translFsFlexvSLH{\annotated{\cexpr{c_1}}}}{\translFsFlexvSLH{\annotated{\cexpr{c_2}}}}\\
    \translFsFlexvSLH{\acifs{\cexpr{be}}{\annotated{\cexpr{c_\BoolTrue}}}{\annotated{\cexpr{c_\BoolFalse}}}{\ell}} &\doteq \texttt{if }\cexpr{\llbracket be \rrbracket_\mathbb{B}^\ell}\\
                                                                                                                   &\phantom{\doteq}\hspace{0.5em} \texttt{then }\cseq{\casgn{\cvar{b}}{\ccond{\cexpr{\llbracket be \rrbracket_\mathbb{B}^\ell}}{\cvar{b}}{1}}}{\translFsFlexvSLH{\cexpr{c_\BoolTrue}}}\\
                                                                                                                   &\phantom{\doteq}\hspace{0.5em}\texttt{else }\cseq{\casgn{\cvar{b}}{\ccond{\cexpr{\llbracket be \rrbracket_\mathbb{B}^\ell}}{1}{\cvar{b}}}}{\translFsFlexvSLH{\cexpr{c_\BoolFalse}}}\\
\translFsFlexvSLH{\acwhiles{\cexpr{be}}{\annotated{\cexpr{c}}}{\ell}{\PubVars}{\PubArrs}} &\doteq \texttt{while }\cexpr{\llbracket be \rrbracket_\mathbb{B}^\ell}\\
                                                                                          &\phantom{\doteq}\hspace{1.5em}\cseq{\casgn{\cvar{b}}{\ccond{\cexpr{\llbracket be \rrbracket_\mathbb{B}^\ell}}{\cvar{b}}{1}}}{\translFsFlexvSLH{\annotated{\cexpr{c}}}}\\
                                                                                              &\phantom{\doteq}\hspace{0.5em}\casgn{\cvar{b}}{\ccond{\cexpr{\llbracket be \rrbracket_\mathbb{B}^\ell}}{1}{\cvar{b}}}\\
    \translFsFlexvSLH{\acaread{\cvar{X}}{\cvar{a}}{\cexpr{i}}{\ell_\cvar{X}}{\ell_i}} &\doteq 
  \left\{
  \begin{matrix*}[l]
  {\cseq{\caread{\cvar{X}}{\cvar{a}}{\cexpr{i}}}{\casgn{\cvar{X}}{\ccond{\cvar{b==1}}{\cvar{0}}{\cvar{X}}}}} &
  {\mathit{if} ~ \ell_\cvar{X} \land \ell_i}
  \\
  {\caread{\cvar{X}}{\cvar{a}}{\llbracket \cexpr{i} \rrbracket_{\mathit{rd}}^{\ell_\cvar{X}, \ell_i}}} &
  {\mathit{otherwise}}
  \end{matrix*}
  \right.\\
    \translFsFlexvSLH{\acawrite{\cvar{a}}{\cexpr{i}}{\cexpr{e}}{\ell_i}} &\doteq \cawrite{\cvar{a}}{\llbracket \cexpr{i} \rrbracket_{\mathit{wr}}^{\ell_i}}{\cexpr{e}}
  \end{align*}
\end{minipage}
\begin{minipage}{0.3\textwidth}
  \begin{align*}
    \llbracket \cexpr{be} \rrbracket_\mathbb{B}^\BoolTrue &\doteq \cexpr{be}\\
    \llbracket \cexpr{be} \rrbracket_\mathbb{B}^\BoolFalse &\doteq \cvar{b} = 0 \land \cexpr{be}\\
    \\
    \llbracket \cexpr{i} \rrbracket_{\mathit{rd}}^{\BoolFalse, \BoolTrue} &\doteq \cexpr{i}\\
    \llbracket \cexpr{i} \rrbracket_{\mathit{rd}}^{\cdot,\cdot} &\doteq \ccond{\cvar{b} = 1}{0}{\cexpr{i}} \quad\text{otherwise}\\
    \\
    \llbracket \cexpr{i} \rrbracket_{\mathit{wr}}^\BoolTrue &\doteq \cexpr{i}\\
    \llbracket \cexpr{i} \rrbracket_{\mathit{wr}}^\BoolFalse &\doteq \ccond{\cvar{b} = 1}{0}{\cexpr{i}}
  \end{align*}
\end{minipage}
  \caption{Translation function for \FsFlexvSLH}
  \label{fig:trans-fsflexvslh}
\end{figure*}

\subsection{Ideal Semantics}
\label{sec:appendix-fsflexvslh-ideal-semantics}
\begin{figure*}
  \centering
  \[
    \inferrule*[left=Ideal\_If]
    {\highlight{\cancel{\labelbool{\PubVars}{\cexpr{be}} = \ell}} \\
   b' = (\ell \vee \neg b) \mathrel{\wedge} \eval{\cexpr{be}}{\ScalarState}
  }
  {
    \idealstep{\fsidealstate{\highlight{\acifs{\cexpr{be}}{\annotated{\cexpr{c_\BoolTrue}}}{\annotated{\cexpr{c_\BoolFalse}}}{\ell}}}{\rho}{\mu}{b}{\highlight{\mathpc}}{\highlight\PubVars}{\highlight{\PubArrs}}}{\branch{b'}}{\step}{\fsidealstate{\highlight{\acbranch{\mathpc}{\annotated{\cexpr{c_{b'}}}}}}{\rho}{\mu}{b}{\highlight{\mathpc \sqcup \ell}}{\highlight{\PubVars}}{\highlight{\PubArrs}}}
}
  \]
  \[
  \inferrule*[left=Ideal\_If\_Force]
  {\highlight{\cancel{\BoolLabel{\cexpr{be}} = \ell}} \\
   b' = (\ell \vee \neg b) \mathrel{\wedge} \eval{\cexpr{be}}{\ScalarState}
  }
  {
    \idealstep{\fsidealstate{\highlight{\acifs{\cexpr{be}}{\annotated{\cexpr{c_\BoolTrue}}}{\annotated{\cexpr{c_\BoolFalse}}}{\ell}}}{\rho}{\mu}{b}{\highlight{\mathpc}}{\highlight{\PubVars}}{\highlight{\PubArrs}}}{\branch{b'}}{\force}{\fsidealstate{\highlight{\acbranch{\mathpc}{\annotated{\cexpr{c_{\neg b'}}}}}}{\rho}{\mu}{\BoolTrue}{\highlight{\mathpc \sqcup \ell}}{\highlight{\PubVars}}{\highlight{\PubArrs}}}
    }
  \]
  \[
  \inferrule*[left=Ideal\_Read]
  {\highlight{\cancel{\labelarith{\PubVars}{\cexpr{ie}} = \ell_{\cexpr{i}}} }\\
   i = \left\{
   {\begin{matrix*}[l]
    0 & \mathit{if} ~ \lnot\ell_{\cexpr{i}} \mathrel{\wedge} b \\
    \eval{\cexpr{ie}}{\rho} & \mathit{otherwise}
    \end{matrix*}
   }
   \right.
   \\
   v =
   \left\{
   {\begin{matrix*}[l]
       0 & \mathit{if} ~ \highlight{\ell_\cvar{X}} \mathrel{\wedge} \ell_{\cexpr{i}} \mathrel{\wedge} b \\
    \colAll{\eval{\cvar{a}[i]}{\mu}} & \mathit{otherwise}
    \end{matrix*}
   }
   \right.
    \\
   i < | \cvar{a} |_{\mu}
 }{
   \idealstep{\fsidealstate{\highlight{\acaread{\cvar{X}}{\cvar{a}}{\cexpr{ie}}{\ell_\cvar{X}}{\ell_\cexpr{i}}}}{\rho}{\mu}{b}{\highlight{\mathpc}}{\highlight{\PubVars}}{\highlight{\PubArrs}}}{\OARead{\cvar{a}}{i}}{\step}{\fsidealstate{\cskip}{\subst{\cvar{X}}{v}{\rho}}{\mu}{b}{\highlight{\mathpc}}{\highlight{\subst{\cvar{X}}{\ell_\cvar{X}}{\PubVars}}}{\highlight{\PubArrs}}}
 }
  \]
  \[
  \inferrule*[left=Ideal\_Read\_Force]
  {\highlight{\cancel{\labelarith{\PubVars}{\cexpr{ie}}}} \\
   i = \eval{\cexpr{ie}}{\rho} \\
   v =
   \left\{
   {\begin{matrix*}[l]
       0 & \mathit{if} ~ \highlight{\ell_\cvar{X}} \\
    \colAll{\eval{\cvar{b}[j]}{\mu}} & \mathit{otherwise}
    \end{matrix*}
   }
   \right.
    \\
   i \ge | \cvar{a} |_{\mu} \\
   j < | \cvar{b} |_{\mu} 
 }{
   \idealstep{\fsidealstate{\highlight{\acaread{\cvar{X}}{\cvar{a}}{\cexpr{ie}}{\ell_\cvar{X}}{\BoolTrue}}}{\rho}{\mu}{\BoolTrue}{\highlight{\mathpc}}{\highlight{\PubVars}}{\highlight{\PubArrs}}}{\OARead{\cvar{a}}{i}}{\DLoad{\cvar{b}}{j}}{\fsidealstate{\cskip}{\subst{\cvar{X}}{v}{\rho}}{\mu}{\BoolTrue}{\highlight{\mathpc}}{\highlight{\subst{\cvar{X}}{\ell_\cvar{X}}{\PubVars}}}{\highlight{\PubArrs}}}
 }
  \]
  \[
  \inferrule*[left=Ideal\_Write]
  {
   i = \left\{
   {\begin{matrix*}[l]
    0 & \mathit{if} ~ \lnot\ell_{\cexpr{i}} \mathrel{\wedge} b \\
    \eval{\cexpr{ie}}{\rho} & \mathit{otherwise}
    \end{matrix*}
   }
   \right. \\
   \highlight{\cancel{\labelarith{\PubVars}{\cexpr{ie}} = \ell_{\cexpr{i}}}} \\
   v = \eval{\cexpr{ae}}{\rho} \\ 
   i < | \cvar{a} |_{\mu}
  }
  {
    \idealstep{\fsidealstate{\highlight{\acawrite{\cvar{a}}{\cexpr{ie}}{\cexpr{ae}}{\ell_\cexpr{ie}}}}{\rho}{\mu}{b}{\highlight{\mathpc}}{\highlight{\PubVars}}{\highlight{\PubArrs}}}{\OAWrite{\cvar{a}}{i}}{\step}{\fsidealstate{\cskip}{\rho}{\subst{\cvar{a}[i]}{v}{\mu}}{b}{\highlight{\mathpc}}{\highlight{\PubVars}}{\highlight{\subst{\cvar{a}}{\PubVars(\cvar{a}) \sqcup \mathpc \sqcup \ell_\cexpr{ie} \sqcup \labelarith{\PubVars}{\cexpr{ae}}}{\PubArrs}}}}
}
   \]
   \[
  \inferrule*[left=Ideal\_Write\_Force]
  {\highlight{\cancel{\labelarith{\PubVars}{\cexpr{ie}}}} \\
   i = \eval{\cexpr{ie}}{\rho} \\
   v = \eval{\cexpr{ae}}{\rho} \\  
   i \ge | \cvar{a} |_{\mu} \\
   j < | \cvar{b} |_{\mu}
  }
  {
    \idealstep{\fsidealstate{\highlight{\acawrite{\cvar{a}}{\cexpr{ie}}{\cexpr{ae}}{\BoolTrue}}}{\rho}{\mu}{\BoolTrue}{\highlight{\mathpc}}{\highlight{\PubVars}}{\highlight{\PubArrs}}}{\OAWrite{\cvar{a}}{i}}{\DStore{\cvar{b}}{j}}{\fsidealstate{\cskip}{\rho}{\subst{\cvar{b}[j]}{v}{\mu}}{\BoolTrue}{\highlight{\mathpc}}{\highlight{\PubVars}}{\highlight{\subst{\cvar{a}}{\PubArrs(\cvar{a}) \sqcup \mathpc \sqcup \labelarith{\PubVars}{\cexpr{ae}}}{\PubArrs}}}}
}
  \]
  \[
  \inferrule*[left=Ideal\_Seq\_Skip]
  {\highlight{\mathit{terminal} ~ \annotated{\cexpr{c_1}}}}
  {\idealstep{\fsidealstate{\highlight{\acseq{\annotated{\cexpr{c_1}}}{{\annotated{\cexpr{c_2}}}}{\PubVars'}{\PubArrs'}}}{\rho}{\mu}{b}{\highlight{\mathpc}}{\highlight{\PubVars}}{\highlight{\PubArrs}}}{\bullet}{\bullet}{\fsidealstate{\annotated{\cexpr{c_2}}}{\rho}{\mu}{b}{\highlight{\pcofacom{\annotated{\cexpr{c_1}}}{\mathpc}}}{\highlight{\PubVars}}{\highlight{\PubArrs}}}}
  \]
  \[
    \highlight{
    \inferrule*[left=Ideal\_Branch]
    {
      \idealstep{\fsidealstate{{\annotated{\cexpr{c}}}}{\rho}{\mu}{b}{\mathpc}{\PubVars}{\PubArrs}}{o}{d}{\fsidealstate{{\annotated{\cexpr{c'}}}}{\rho'}{\mu'}{b'}{\mathpc'}{\PubVars'}{\PubArrs'}}
    }
    {
      \idealstep{\fsidealstate{\acbranch{\ell}{\annotated{\cexpr{c}}}}{\rho}{\mu}{b}{\mathpc}{\PubVars}{\PubArrs}}{o}{d}{\fsidealstate{\acbranch{\ell}{\annotated{\cexpr{c'}}}}{\rho'}{\mu'}{b'}{\mathpc'}{\PubVars'}{\PubArrs'}}
    }
  }
  \]
  \caption{Ideal semantics for \FsFlexvSLH}
  \label{fig:ideal-semantics-flowsensitivevslh-full}
  \jb{someone please proofread these for me?}
  \jb{this actually still isn't all rules}
\end{figure*}
\autoref{fig:ideal-semantics-flowsensitivevslh-full} shows the full ideal semantics for \FsFlexvSLH.
Highlighted components still indicate changes relative to the ideal semantics for \FlexvSLH (\autoref{fig:ideal-semantics-fvslh}).
\else
\fi 

\ifieee
\bibliographystyle{abbrvnaturl}
\else 
\ifcamera
\bibliographystyle{ACM-Reference-Format}
\citestyle{acmauthoryear}   
\else
\bibliographystyle{abbrvnaturl}
\fi

\fi 

\bibliography{spec}

\end{document}